\documentclass{llncs}
\usepackage[utf8]{inputenc}
\usepackage{graphicx}

\usepackage{times}
\usepackage{soul}
\usepackage{url}
 
\usepackage[utf8]{inputenc}
\usepackage[small]{caption}
\usepackage{graphicx}
\usepackage{amsmath}
 
\usepackage{booktabs}
\usepackage{bm}
\usepackage{stmaryrd}

\usepackage{stackrel}
\usepackage{amssymb}
\usepackage{amsmath}
\usepackage{mathrsfs}
\usepackage{comment}
\usepackage{todonotes}
\usepackage[inline]{enumitem}
\usepackage{float}
\usepackage{algorithmic}
\usepackage{framed}
\usepackage[linesnumbered,lined,commentsnumbered]{algorithm2e}
\usepackage{amsmath,commath}

\usepackage{caption}
\usepackage{longtable}

\newcommand{\tuple}[1]{\langle #1 \rangle}

\newcommand{\Catc}{\mathscr{C}_c}
\mathchardef\mhyphen="2D 

\tikzset{
  small dot/.style={
    circle, inner sep=0pt, 
    minimum size=1mm, fill=black
  }
}

\tikzset{
  big dot/.style={
    circle, inner sep=0pt, 
    minimum size=2mm, fill=black
  }
}

\newlist{enumaa}{enumerate*}{5}
\setlist[enumaa]{label={(\emph{\alph*})}}

\newcommand{\leftrarrows}{\mathrel{\raise.75ex\hbox{\oalign{%
  $\scriptstyle\leftarrow$\cr
  \vrule width0pt height.5ex$\hfil\scriptstyle\relbar$\cr}}}}
\newcommand{\lrightarrows}{\mathrel{\raise.75ex\hbox{\oalign{%
  $\scriptstyle\relbar$\hfil\cr
  $\scriptstyle\vrule width0pt height.5ex\smash\rightarrow$\cr}}}}
\newcommand{\Rrelbar}{\mathrel{\raise.75ex\hbox{\oalign{%
  $\scriptstyle\relbar$\cr
  \vrule width0pt height.5ex$\scriptstyle\relbar$}}}}
\newcommand{\longleftrightarrows}{\leftrarrows\joinrel\Rrelbar\joinrel\lrightarrows}

\begin{document}
\title{Reasoning in the Description Logic $\mathcal{ALC}$ under Category Semantics}
\author{Ludovic Brieulle  \and Chan {Le Duc} \and Pascal Vaillant }

\institute{Universit\'e Sorbonne Paris Nord, LIMICS, INSERM, U1142, F-93000, Bobigny, France\\
%\email{firstname.lastname@univ-paris13.fr}
}

\maketitle
\begin{abstract}
 We present in this paper a reformulation of the usual set-theoretical semantics of the description logic  $\mathcal{ALC}$ with general TBoxes  by using categorical language. In this setting, $\mathcal{ALC}$ concepts are represented as objects, concept subsumptions as  arrows, and memberships as logical quantifiers over objects and arrows of  categories. 
Such a category-based semantics provides a more modular representation of the semantics of $\mathcal{ALC}$. This feature allows us to define a sublogic of $\mathcal{ALC}$ by dropping the interaction between existential and universal restrictions, which would be responsible for an exponential complexity in space. Such a sublogic is undefinable in the usual set-theoretical semantics.  We show that  this sublogic  is {\sc{PSPACE}}  by proposing a deterministic algorithm for checking concept satisfiability which runs  in polynomial space.
\end{abstract}

\section{Introduction} \label{sec:intro}
Languages based on Description Logics (DLs) \cite{baa10} such as OWL \cite{pat04}, OWL2 \cite{grau08}, are widely used to represent  ontologies   in semantics-based applications. $\mathcal{ALC}$ is the smallest DL involving roles which is closed under negation. It is a suitable logic for a first attempt to replace the usual set-theoretical semantics by another one.  A pioneer work by Lawvere \cite{law64} provided  an appropriate axiomatization of the category of sets by replacing set membership with the composition of functions. However, it was not indicated whether the categorical axioms are ``semantically"  equivalent to the axioms based on set membership. As pointed out by Goldblatt \cite{gol06}, this may lead to a very different semantics for negation.   In this paper,  we propose a rewriting of  the usual set-theoretical  semantics of $\mathcal{ALC}$ by using objects and arrows, the two fundamental  elements of category theory, to represent concepts and subsumptions respectively.

It is well known that  $\mathcal{ALC}$ with general TBoxes is {\sc{EXPTIME}}-complete \cite{schild1991,baader2017} while $\mathcal{ALC}$ without   TBox is {\sc{PSPACE}}-complete \cite{Schmidt1991}. Hence, the interaction between existential and universal restrictions with general TBoxes would be responsible for an intractable space complexity if {\sc{PSPACE}}$\subsetneq ${\sc{EXPTIME}}. The main motivation of this work consists in  identifying a sublogic of $\mathcal{ALC}$ with general TBoxes which allows to reduce the reasoning complexity without losing too much of the expressive power.  

To reduce space complexity when reasoning with $\mathcal{ALC}$ under the usual set-theoretical semantics, one can restrict expressiveness of TBoxes by preventing them from having cyclic axioms~\cite{baa2008}. This restriction may be too strong for those who wish to express simple knowledge of cyclic nature such as $\mathsf{Human} \sqsubseteq \exists \mathsf{hasParent}.\mathsf{Human}$.

However, we can find a sublogic of $\mathcal{ALC}$ under category-theoretical semantics such that it allows to fully express general TBoxes and only needs to drop a part of the semantics of universal restrictions. Such a sublogic is undefinable in the usual set-theorectical semantics. Indeed, a universal restriction $\forall R.C$ in  $\mathcal{ALC}$ can be defined under category-theoretical semantics by using the following two informal properties (that will be developed in more detail below, Section~\ref{sec:cat-semantics}, Definition~\ref{def:forall}): (i) $\forall R.C$ is ``very close'' to $\neg \exists R.\neg C$; and (ii)    $\exists R.C\sqcap \forall R.D$ is ``very close'' to $\exists R.(C\sqcap  D)$.

We can observe that Property~(ii) is a (weak) representation of the interaction between universal and existential restrictions.

As will be shown below (Section~\ref{sec:reasoning}), if we may just remove this interaction from the categorical semantics of universal restriction, we obtain a new logic, namely $\mathcal{ALC}_{\overline{\forall}}$, in which reasoning will be tractable in space.

The semantic loss caused by this removal might be tolerable in certain cases. We illustrate it with the example below.

\begin{example}
  Consider for instance the following $\mathcal{ALC}$ TBox:
  
  \begin{tabular}{l}
      $\mathsf{HappyChild} \sqsubseteq \exists \mathsf{eatsFood}.\mathsf{Dessert} \sqcap  \forall \mathsf{eatsFood}. \mathsf{HotMeal}$ \\
       $\mathsf{Dessert} \sqsubseteq \neg \mathsf{HotMeal}$
\end{tabular}

\noindent  As in the usual set-theoretical semantics, nobody can be a $\mathsf{HappyChild}$ under category-theoretical semantics: the concept is unsatisfiable in this world. Indeed, according to the first axiom if somebody is a $\mathsf{HappyChild}$, then they must simultaneously (1)~have some $\mathsf{Dessert}$ to eat, and (2)~eat only $\mathsf{HotMeal}$, which contradicts the second axiom.

Now, if the second axiom is removed from the TBox, under set-theoretical semantics, the first axiom entails that if somebody is a $\mathsf{HappyChild}$, then they eat $\mathsf{HotMeal}$. In fact, the first subconcept ($\exists \mathsf{eatsFood}.\mathsf{Dessert}$) ensures that there exists at least one food item that they eat, and the second one ($\forall \mathsf{eatsFood}. \mathsf{HotMeal}$) that this food is necessarily $\mathsf{HotMeal}$.

Under the category-theoretical semantics that we will define in this paper ($\mathcal{ALC}_{\overline{\forall}}$), the first axiom alone does not allow to entail that if somebody is a $\mathsf{HappyChild}$ then they eat $\mathsf{HotMeal}$. This is due to the fact that the element of the definition of universal quantification that was there to represent Property~(ii) has been dropped, and that unlike set-theoretical semantics, we no longer have a set of individuals providing an extensional ``support'' for the second subconcept.
\end{example}

The paper is organized as follows. We  begin by translating  semantic constraints related to each $\mathcal{ALC}$ constructor into arrows between objects. Then, we check whether the obtained arrows allow to restore usual properties. If it is not the case, we add  new arrows and objects to capture the missing properties without going beyond the set-theoretical semantics. For instance, it is not sufficient to define category-theoretical semantics of negation $\neg C$ by stating $C \sqcap \neg C \longrightarrow \bot$ and $\top \longrightarrow C \sqcup \neg C$ because it is not possible to obtain  arrows such as  $C \longleftrightarrows \neg \neg C$ from this definition \cite{leduc2021}. In Section~\ref{sec:reasoning}, we begin by identifying a sublogic of $\mathcal{ALC}$, namely $\mathcal{ALC}_{\overline{\forall}}$, by dropping a property from the categorical semantics of universal restriction which would lead to an intractable complexity in space. We show that this sublogic is strictly different from  $\mathcal{ALC}$ and propose a {\sc{PSPACE}} deterministic algorithm for checking concept satisfiability in  $\mathcal{ALC}_{\overline{\forall}}$.

\section{Related Work}\label{sec:tableau-connection}
In this section, we discuss some results on {\sc{PSPACE}} algorithms for Description Logics that are slightly more expressive than $\mathcal{ALC}$.  A {\sc{PSPACE}} algorithm \cite{hor99} was presented for the logic  $\mathcal{SI}$ ($\mathcal{ALC}$ with  inverse and transitive roles) without TBox. Tableau method is used in this algorithm to build a tree to represent a model. Since TBox is not allowed, the depth of such trees is bounded by a polynomial function in the size of input. When extending this method to a general TBox, the depth of such trees may be exponential. A more recent work \cite{baa2008} proposed an automata-based algorithm for $\mathcal{SI}$ with an acyclic TBox. This algorithm tries to build a \emph{tree automaton} to represent a model.  In this construction, the acyclicity of a TBox prevents the algorithm from building a tree automaton with an exponential depth.         

A common issue of these algorithms is that they need to store backtracking points in order to deal with nondeterminism. When the depth is exponential, they have to use an exponential memory to store backtracking points. 

There have been very few works on connections between category theory and DLs. Spivak et al.~\cite{spivak2012} used  category theory to define a high-level graphical language comparable with OWL (based on DLs) for knowledge representation rather than a foundational formalism for reasoning.
%\cite{goguen1992}
Codescu \textit{et al.}~\cite{CodescuMK17} introduced a categorical framework to formalize a network of aligned ontologies. The formalism on which  the framework is based is independent from the logic used in ontologies. It is shown that all global reasoning over such a network can be reduced to local reasoning for a given DL used in the ontologies involved in the network. As a consequence, the semantics of DLs employed in the ontologies continue to rely on set theory.

\section{Syntax and set-theoretical semantics of $\mathcal{ALC}$} \label{sec:alc}

We present syntax and semantics of the Description Logic $\mathcal{ALC}$ \cite{baa10} with  TBoxes and some basic inference  problems.

\begin{definition}[Syntax and set-theoretical  semantics]\label{def:alc}
Let  $\mathbf{R}$ and  $\mathbf{C}$   be non-empty sets of \emph{role names} and \emph{concept names}    respectively. 
The set of $\mathcal{ALC}$-concepts is defined as the smallest set containing all concept names  in $\mathbf{C}$ with  $\top$, $\bot$ and complex concepts that are inductively defined as follows: $C\sqcap D$, $C\sqcup D$, $\neg C$, $\exists R.C$, $\forall R.C$ where  $C$ and $D$ are $\mathcal{ALC}$-concepts, and $R$ is a role name in $\mathbf{R}$. An axiom $C\sqsubseteq D$ is called a general concept inclusion (GCI) where $C,D$ are (possibly complex) $\mathcal{ALC}$-concepts. An  $\mathcal{ALC}$ ontology  $\mathcal{O}$ (or general TBox) is a finite set of GCIs. 
 
 An interpretation ${\cal I}=\langle\Delta^{\cal{I}},\cdot^{\cal{I}}\rangle$ consists of a non-empty set $\Delta^{\cal{I}}$ ({\em domain}), and a function $\cdot^{\cal{I}}$ ({\em interpretation function}) which  associates a subset of $\Delta^{\cal{I}}$ to each concept name, an element  in $\Delta^\mathcal{I}$ to each individual, and a subset of $\Delta^{\cal{I}}\times \Delta^{\cal{I}}$ to each role name such that  
\begin{center}
\begin{tabular}{@{}l@{}ll@{}} %{lll}
$\top^{\cal{I}}$ & $=$ & $\Delta^{\cal{I}}$\\ 
$\bot^{\cal I}$ &$=$ & $\emptyset$\\
$~(C\sqcap D)^{\cal I}$ &$=$ & $C^{\cal I}\cap D^{\cal I}$\\   
$(C\sqcup D)^{\cal I}$ &$=$ & $C^{\cal I}\cup D^{\cal I}$\\
$(\neg C)^{\cal I}$ &$=$ & $\Delta^{\cal{I}}\setminus C^{\cal I}$\\
$(\exists R.C)^{\cal I}$ &$=$ & $\{x\!\in\!\Delta^{\cal{I}}\!\mid\!\exists y\!\in\!\Delta^{\cal{I}},\!\langle x,y\rangle\!\in\!{R^{\cal I}}\!\wedge\!y\!\in\! C^{\cal I}\}$\\
$(\forall R.C)^{\cal I}$ &$=$ & $\{x\!\in\!\Delta^{\cal{I}}\!\mid\! \forall y\!\in\!\Delta^{\cal{I}},\langle x,y\rangle\!\in\! {R^{\cal I}}\!\implies \! y\!\in\! C^{\cal I}\}$\\
\end{tabular}  
\end{center}

An interpretation $\mathcal{I}$ satisfies  a GCI $C\sqsubseteq D$  if $C^{\cal I}\subseteq D^{\cal I}$.
%(ii) an assertion $C(a)$  if $a^\mathcal{I} \in  C^{\cal I}$, (iii) an assertion $R(a,b)$  if $\langle a^\mathcal{I}, b^\mathcal{I} \rangle \in  R^{\cal I}$. 
$\mathcal{I}$ is a model of  $\mathcal{O}$, written $\mathcal{I}\models  \mathcal{O}$, if   $\mathcal{I}$ satisfies  each GCI of $\mathcal{O}$.   In this case, we say that $\mathcal{O}$ is {\em set-theoretically consistent}, and {\em set-theoretically inconsistent} otherwise. A concept $C$ is {\em set-theoretically satisfiable} with respect to $\mathcal{O}$ if there is a model $\mathcal{I}$ of $\mathcal{O}$ such that $C^\mathcal{I}\neq \emptyset$, and {\em set-theoretically unsatisfiable} otherwise. We say that a GCI $C\sqsubseteq D$ is {\em set-theoretically entailed} by $\mathcal{O}$, written $\mathcal{O}\models C\sqsubseteq D$, if $C^\mathcal{I}\subseteq D^\mathcal{I}$ for all models $\mathcal{I}$ of $\mathcal{O}$. The pair $\langle\mathbf{C},  \mathbf{R}\rangle$  is called the signature of $\mathcal{O}$.
\end{definition}

We finish this section by introducing notations  that will be used in the paper. We use $\lvert S \rvert$ to denote the cardinality of a set $S$. Given an $\mathcal{ALC}$ ontology $\mathcal{O}$, we denote by $\mathsf{sub}(\mathcal{O})$ the set of all sub-concepts occurring in   $\mathcal{O}$. The size of an ontology $\mathcal{O}$, denoted by $\norm{\mathcal{O}}$,  is   the size of all GCIs. Analogously, we use $\norm{C}$ to denote the size of a concept $C$. It holds that $\lvert \mathsf{sub}(\mathcal{O})\lvert$ is polynomially bounded by $\norm{ \mathcal{O}}$ since if a concept is represented as string then a sub-concept is a substring.

\section{Category-theoretical semantics of $\mathcal{ALC}$} \label{sec:cat-semantics}

We can observe that the set-theoretical semantics of $\mathcal{ALC}$ is based on  set membership relationships, while ontology inferences, such as consistency or concept subsumption, involve set inclusions. This explains why inference algorithms developed in this setting often have to build sets of individuals connected in some way for representing a model.

In this section, we use some basic notions in category theory to characterize the semantics of $\mathcal{ALC}$. Instead of set membership, in this categorical language, we use ``objects'' and ``arrows'' to define semantics of a given object.

%\pv{Perhaps you should re-read carefully the paragraph below, that I have written in order for the reader to find the path into abstraction a little less steep. If you think it is somehow ``betraying'' the logic of the paper, we will get rid of it.}

%Our approach consists in introducing the syntax of this categorical version of $\mathcal{ALC}$, exposing the mechanism and internal logic of its system, and then progressively showing that it captures exactly the same semantics as the set-conceptual version.  This being said, and with no intention of ``spoiling'' the section that follows, we may still give a few hints to the readers so that they see where we are trying to get at.  Here there are two types of ``objects'' corresponding respectively to concepts and roles, and ``arrows'' are another way to express concept inclusions. The connection between concepts and roles is no longer defined in terms of set relations, but with the help of two functors $\mathsf{dom}$ and $\mathsf{cod}$, which are there to capture respectively the conceptual contents of the {\em domain} and {\em codomain} of a role.

%\cld{This does not fit because we never introduce category syntax. Moreover, we shouldn't use two different terms (set-conceptual version and set-theoretical ...) for a unique meaning}

Although the present paper is self-contained, we refer the readers to textbooks   \cite{gol06,saunders92} on category theory for further information.

\begin{definition}[Syntax categories]\label{def:syntax-cat}
Let  $\mathbf{R}$ and  $\mathbf{C}$   be non-empty sets of \emph{role names} and \emph{concept names}    respectively. We define a \emph{concept syntax category} $\mathscr{C}_c$ and a \emph{role syntax category} $\mathscr{C}_r$ from the signature $\langle\mathbf{C},  \mathbf{R}\rangle$  as follows:
 
\begin{enumerate}
   
\item\label{def:syntax-cat:1} Each role name $R$ in  $\mathbf{R}$ is an object $R$  of   $\mathscr{C}_r$. In particular, there are initial and terminal objects $R_\bot$ and $R_\top$ in     $\mathscr{C}_r$ with arrows $R\longrightarrow R_\top$ and $R_\bot\longrightarrow R$   for all object $R$ of  $\mathscr{C}_r$. There is also an identity arrow $R\longrightarrow R$ for each object $R$  of  $\mathscr{C}_r$.
 
    \item\label{def:syntax-cat:2} Each concept name in $\mathbf{C}$  is an object of $\mathscr{C}_c$. In particular, $\bot$ and $\top$  are respectively initial and terminal objects, i.e. there are arrows $C\longrightarrow \top$ and $\bot \longrightarrow C$ for each object $C$ of  $\mathscr{C}_c$. Furthermore, for each object  $C$  of  $\mathscr{C}_c$   there is an identity arrow $C\longrightarrow C$, 
    and for each object  $R$  of  $\mathscr{C}_r$ there  are two objects of $\mathscr{C}_c$, namely $\mathsf{dom}(R)$ and $\mathsf{cod}(R)$.

\item\label{def:syntax-cat:3} If there are arrows $E\longrightarrow F$ and $F\longrightarrow G$ in $\mathscr{C}_c$ (resp. $\mathscr{C}_r$), then there is an arrow $E\longrightarrow G$ in $\mathscr{C}_c$ (resp. $\mathscr{C}_r$).

\item \label{def:syntax-cat:4}There are two functors $\mathsf{dom}$ and $\mathsf{cod}$ from  $\mathscr{C}_r$ to $\mathscr{C}_c$, i.e.  they associate two objects $\mathsf{dom}(R)$ and $\mathsf{cod}(R)$ of $\mathscr{C}_c$ to each object $R$ of $\mathscr{C}_r$ such that  
%identity arrows and arrow compositions are preserved from  $\mathscr{C}_r$ to $\mathscr{C}_c$ as follows:
\begin{enumerate}
\item $\mathsf{dom}(R_\top)=\top$, $\mathsf{cod}(R_\top)=\top$, $\mathsf{dom}(R_\bot)=\bot$ and $\mathsf{cod}(R_\bot)=\bot$.
    \item 
if there is an arrow $R\longrightarrow R'$ in $\mathscr{C}_r$ then there are arrows $\mathsf{dom}(R)\longrightarrow \mathsf{dom}(R')$ and $\mathsf{cod}(R)\longrightarrow \mathsf{cod}(R')$.%  in $\mathscr{C}_c$, 
\item if there are arrows $R\longrightarrow R'\longrightarrow R''$ in $\mathscr{C}_r$ then there are arrows $\mathsf{dom}(R)\longrightarrow \mathsf{dom}(R'')$ and $\mathsf{cod}(R)\longrightarrow \mathsf{cod}(R'')$.

\item if there is an arrow $\mathsf{dom}(R)\longrightarrow \bot$ or  $\mathsf{cod}(R)\longrightarrow \bot$ in $\mathscr{C}_c$, then there is  an arrow $R\longrightarrow R_\bot$ in $\mathscr{C}_r$. 
\end{enumerate}

\end{enumerate}
For each arrow $E\longrightarrow F$ in $\mathscr{C}_c$ or $\mathscr{C}_r$, $E$ and $F$ are respectively  called \emph{domain} and \emph{codomain} of the  arrow. We  use also $\mathsf{Ob}(\mathscr{C})$ and $\mathsf{Hom}(\mathscr{C})$  to denote the collections of objects and arrows of a category $\mathscr{C}$.

\end{definition}

Definition~\ref{def:onto-cat} provides a general framework with syntax elements and necessary properties  from category theory. We need to ``instantiate" it to obtain  categories which capture  semantic constraints  coming from axioms. The following definition extends syntax categories in such a way that they admit the axioms of an $\mathcal{ALC}$  ontology as  arrows.

\begin{definition}[Ontology categories]\label{def:onto-cat}
Let $C$ be an $\mathcal{ALC}$ concept and $\mathcal{O}$ an $\mathcal{ALC}$ ontology from a  signature $\langle \mathbf{C}, \mathbf{R}\rangle$.  We define a \emph{concept ontology category} $\mathscr{C}_c\langle C, \mathcal{O}\rangle$ and a \emph{role ontology category} $\mathscr{C}_r\langle C, \mathcal{O}\rangle$  as follows:
\begin{enumerate}

\item $\mathscr{C}_c\langle C, \mathcal{O}\rangle$ and  $\mathscr{C}_r\langle C, \mathcal{O}\rangle$ are syntax categories from  $\langle \mathbf{C}, \mathbf{R}\rangle$.

\item $C$ is an object of $\mathscr{C}_c\langle C, \mathcal{O}\rangle$. 
 
\item If $E \sqsubseteq F$ is an axiom of $\mathcal{O}$, then $E, F$ are objects and $E\longrightarrow F$ is an arrow of  $\mathscr{C}_c\langle C, \mathcal{O}\rangle$.

%\item If $E \sqsubseteq F$ is an axiom of $\mathcal{O}$, then $\neg E \sqcup F$ is an object of  $\mathscr{C}_c\langle C, \mathcal{O}\rangle$ and there is an arrow $\top \longrightarrow \neg E \sqcup  F$  in $\mathscr{C}_c\langle C, \mathcal{O}\rangle$.

\end{enumerate}
\end{definition}

In this paper, an object of $\mathscr{C}_c\langle C, \mathcal{O}\rangle$ and $\mathscr{C}_r\langle C, \mathcal{O}\rangle$ is called concept and role object respectively. We transfer  the vocabulary used in Description Logics to  categories as follows. A concept object $C\sqcup D$, $C\sqcap D$ or $\neg C$ is respevtively called disjunction, conjunction and negation object. For an existential  restriction object $\exists R.C$ or universal  restriction object $\forall R.C$, $C$ is called the \emph{filler} of $\exists R.C$ and $\forall R.C$.

In the sequel,  we introduce category-theoretical semantics of  disjunction,  conjunction, negation, existential and universal restriction objects if they appear in $\mathscr{C}_c\langle C, \mathcal{O}\rangle$.
 Some of them require more \emph{explicit} properties than those needed for the set-theoretical semantics. This is due to the fact that set membership is translated into arrows in a syntax category.  Since semantics of an object in a category depends to relationships with another ones, the following definitions need to add to $\mathscr{C}_c\langle C, \mathcal{O}\rangle$ and $\mathscr{C}_r\langle C, \mathcal{O}\rangle$ new objects and arrows. 

\begin{definition}[Category-theoretical semantics of  disjunction]\label{def:disj}
Let $C, D, C \sqcup D$ be   concept objects of $\mathscr{C}_c\langle C, \mathcal{O}\rangle$.  Category-theoretical semantics of $C \sqcup D$ is defined by using arrows in $\mathscr{C}_c\langle C, \mathcal{O}\rangle$ as follows. There are arrows $i,j$ from $C$ and $D$ to $C\sqcup D$, and if there is an object $X$ and arrows $f,g$ from $C,D$ to $X$, then there is  an arrow $k$ such that  the following diagram commutes  :

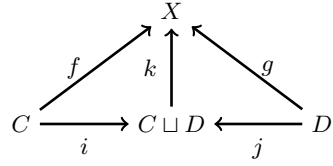
\begin{figure}[!htbp]
\centering
\begin{tikzpicture}[line width=1pt] 
\node[] (AB) at (0,-1) [label={[xshift=-0.3cm, yshift=-0.6cm]}] {$C \sqcup D$};
\node[] (A) at (-2,-1) [label={[xshift=-0.3cm, yshift=-0.6cm]}] {$C$};
\node[] (B) at (2,-1) [label={[xshift=-0.3cm, yshift=-0.6cm]}] {$D$};
\node[] (X) at (0,0.5) [label={[xshift=-0.3cm, yshift=-0.6cm]}] {$X$};

 \draw[->] (A) to[]node[below=2pt]{${i}$} (AB);
 \draw[->] (B) to[]node[below=2pt]{${j}$} (AB);
 \draw[->] (A) to[]node[left=2pt]{${f}$} (X);
 \draw[->] (B) to[]node[right=2pt]{${g}$} (X);
 \draw[->] (AB) to[]node[left=2pt]{${k}$} (X); 
\end{tikzpicture}
\caption{Commutative diagram for disjunction}\label{fig:disj}
\end{figure}

The diagram in Figure~\ref{fig:disj} can be rephrased  as follows:
\begin{align}
&C \longrightarrow C \sqcup D \text{ and } D  \longrightarrow  C \sqcup D   \label{disj01} \\
& \forall X, C    \longrightarrow  X \text{ and } D \longrightarrow  X  \Longrightarrow  C \sqcup D  \longrightarrow X \label{disj02}  
\end{align} 
\end{definition}
 
Intuitively speaking, Properties~(\ref{disj01}) and (\ref{disj02}) tell us that  $C\sqcup D$ is the ``smallest" object which is ``greater" than $C$ and $D$. The following lemma establishes the connection between the usual set semantics of disjunction and the category-theoretical one given in Definition~\ref{def:disj}. In this lemma, Properties~(\ref{disj1}) and (\ref{disj2}) are rewritings of Properties~(\ref{disj01}) and (\ref{disj02}) in set theory. 

\begin{lemma}\label{lem:disj} The category-theorectical semantics of $C\sqcup D$ characterized by  Definition~\ref{def:disj} is compatible with the set-theoretical semantics of $C\sqcup D$, that means
if  $\langle \Delta^\mathcal{I}, \cdot^\mathcal{I}\rangle$ is an interpretation   under  set-theoretical semantics, then the following holds:

$(C\sqcup D)^\mathcal{I}=C^\mathcal{I}\cup D^\mathcal{I}$  iff
\begin{align}
&C^\mathcal{I} \subseteq (C  \sqcup D)^\mathcal{I  }, D^\mathcal{I}  \subseteq  (C  \sqcup D)^\mathcal{I}  \label{disj1}   \\
& \forall X\subseteq \Delta^\mathcal{I}, C^\mathcal{I}    \subseteq  X,   D^\mathcal{I}  \subseteq  X  \Longrightarrow  (C \sqcup D)^\mathcal{I}  \subseteq X   \label{disj2}  
\end{align}
\end{lemma}
\begin{proof}
 \noindent ``$\Longleftarrow$". Due to (\ref{disj1}) we have $C^\mathcal{I}\cup D^\mathcal{I} \subseteq (C\sqcup D)^\mathcal{I}$. Let   $X=C^\mathcal{I}\cup D^\mathcal{I}$. Due to (\ref{disj2}) we have  $(C\sqcup D)^\mathcal{I} \subseteq X=C^\mathcal{I}\cup D^\mathcal{I}$. 
 
 \noindent ``$\Longrightarrow$".  From $(C\sqcup D)^\mathcal{I}=C^\mathcal{I}\cup D^\mathcal{I}$, we have (\ref{disj1}).  Let $x\in (C\sqcup D)^\mathcal{I}$. Due to $(C\sqcup D)^\mathcal{I}=C^\mathcal{I}\cup D^\mathcal{I}$, we have  $x\in C^\mathcal{I}$ or $x\in D^\mathcal{I}$. Hence, $x\in X$  since $C^\mathcal{I}    \subseteq  X$ and  $D^\mathcal{I}  \subseteq  X$. 
\end{proof}
 
At first glance, one can follow the same idea used in Definition~\ref{def:disj}   to define  category-theoretical semantics of $C\sqcap D$ as described in Figure~\ref{fig:conj}. They tell us that   $C\sqcap D$ is the ``greatest" object which is ``smaller" than $C$ and $D$.

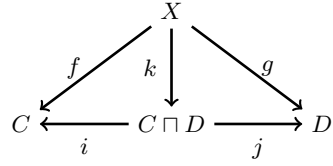
\begin{figure}[!htbp]
\centering
\begin{tikzpicture}[line width=1pt] 
\node[] (AB) at (0,-1) [label={[xshift=-0.3cm, yshift=-0.6cm]}] {$C \sqcap D$};
\node[] (A) at (-2,-1) [label={[xshift=-0.3cm, yshift=-0.6cm]}] {$C$};
\node[] (B) at (2,-1) [label={[xshift=-0.3cm, yshift=-0.6cm]}] {$D$};
\node[] (X) at (0,0.5) [label={[xshift=-0.3cm, yshift=-0.6cm]}] {$X$};

 \draw[->] (AB) to[]node[below=2pt]{${i}$} (A);
 \draw[->] (AB) to[]node[below=2pt]{${j}$} (B);
 \draw[->] (X) to[]node[left=2pt]{${f}$} (A);
 \draw[->] (X) to[]node[right=2pt]{${g}$} (B);
 \draw[->] (X) to[]node[left=2pt]{${k}$} (AB); 
\end{tikzpicture}
\caption{Commutative diagram for a  weak conjunction}\label{fig:conj}
\end{figure}

 However, this definition is not strong enough to entail the
distributive property  of conjunction over disjunction. Hence, we need a stronger definition for conjunction.
%(cf. Example~\ref{ex:distrib})

\begin{definition}[Category-theoretical semantics of  conjunction]\label{def:conj}
 
Let $C, D, E, C \sqcap D, C \sqcap E$, $D \sqcup E$, $C\sqcap (D \sqcup E)$  and $(C \sqcap D)\sqcup  (C \sqcap E)$ be   objects of $\mathscr{C}_c\langle C, \mathcal{O}\rangle$. Category-theoretical semantics of $C \sqcap D$ is defined by using the following properties in $\mathscr{C}_c\langle C, \mathcal{O}\rangle$. 
\begin{align}
&C \sqcap D\longrightarrow  C \text{ and } C \sqcap D  \longrightarrow  D   \label{conj01} \\
%& \forall X, X    \longrightarrow  C \text{ and } X \longrightarrow  D \sqcup E  \Longrightarrow X  \longrightarrow (C \sqcap D)\sqcup (C\sqcap E)\label{conj02} 
& \forall X, X  \longrightarrow  C \text{ and } X \longrightarrow  D    \Longrightarrow  X  \longrightarrow C \sqcap D \label{conj1}\\
& C \sqcap (D \sqcup E) \longrightarrow (C\sqcap D)\sqcup  (C\sqcap E)\label{conj2} 
\end{align} 
\end{definition}

\begin{lemma}
\label{lem:comas}
The following properties hold:
\begin{align}
    &C\sqcup D\longleftrightarrows D\sqcup C \label{com01}\\
    &C\sqcap D\longleftrightarrows D\sqcap C \label{com02}\\
    &(C\sqcup D)\sqcup E\longleftrightarrows C\sqcup(D\sqcup E) \label{asso01}\\
    &(C\sqcap D)\sqcap E\longleftrightarrows C\sqcap(D\sqcap E) \label{asso02}
\end{align}
In other words, $\sqcup$ and $\sqcap$ are commutative and associative with relation to
the objects.
\end{lemma}
\begin{proof}
These properties are direct consequences of the previous definitions.

\begin{enumerate}[wide, labelwidth=!, labelindent=0pt]
\item We start by proving Property (\ref{com01}). Consider object $D\sqcup C$: from Property (\ref{disj01}), we have ${D\longrightarrow D\sqcup C}$ and ${C\longrightarrow D\sqcup C}$. Then, recall Property (\ref{disj02}) above, which states that for all object $X$ such that $C \longrightarrow  X$ and $D \longrightarrow  X$ then  $C \sqcup D\longrightarrow X$. Replace $X$ by $D\sqcup C$ to obtain ${C\sqcup
D\longrightarrow D\sqcup C}$. We obtain arrow ${D\sqcup C\longrightarrow C\sqcup D}$ by 
switching the roles of $D\sqcup C$ and ${C\sqcup D}$ from the above reasoning.

\item Property (\ref{com02}): use the first part of the conjunction definition (\ref{conj01}) on $D\sqcap C$ to obtain ${D\sqcap C\longrightarrow D}$ and ${D\sqcap C\longrightarrow C}$.  From Property~(\ref{conj1}), we know that for all object $X$ such that ${X\longrightarrow C}$ and ${X\longrightarrow D}$ then ${X\longrightarrow C\sqcap D}$. Thus replace $X$ by $D\sqcap C$ to get ${D\sqcap C\longrightarrow C\sqcap D}$. As above, the inverse arrow is obtained by swapping the roles of $D\sqcap C$ and $C\sqcap D$.

\item Property (\ref{asso01}): We start by proving that
\[
(C\sqcup D)\sqcup E\longrightarrow C\sqcup(D\sqcup E).
\]
Let $X = C\sqcup(D\sqcup E)$, we have $C\longrightarrow X$ by (\ref{disj01}),
as well as ${D\longrightarrow D\sqcup E\longrightarrow X}$, then by 
(\ref{disj02}), we get ${C\sqcup D \longrightarrow X}$. Noting we also have 
${E\longrightarrow D\sqcup E\longrightarrow X}$ by definition, we can conclude that 
\[
(C\sqcup D)\sqcup E\longrightarrow X.
\]
The other direction can be proven with no loss of generality by changing the order by which we apply the arrows.

\item Property (\ref{asso02}): As above, proving one direction is virtually the same as proving the other one by swapping the role of $C\sqcap D$ and $D\sqcap E$. Thus we just need to prove the following property:
\[
    C\sqcap(D\sqcap E)\longrightarrow (C\sqcap D)\sqcap E
\]
Let $X = C\sqcap(D\sqcap E)$. From (\ref{conj01}), we get 
${X\longrightarrow C}$ and ${X\longrightarrow D\sqcap E\longrightarrow D}$.
Applying (\ref{conj1}), we obtain $X\longrightarrow C\sqcap D$.
Property~(\ref{conj01}) also gives us ${X\longrightarrow D\sqcap 
E\longrightarrow E}$; so, applying (\ref{conj1}) once again, we end up
with:
\[
X\longrightarrow (C\sqcap D)\sqcap E
\]
which concludes the proof.
\end{enumerate}
\end{proof}

\begin{lemma}
\label{lem:trans}
Assume that we have $D\longrightarrow E$, then the following properties hold:
\begin{align}
    C\sqcup D&\longrightarrow C\sqcup E \label{trans01}\\
    C\sqcap D&\longrightarrow C\sqcap E. \label{trans02}
\end{align}
\end{lemma}
\begin{proof}
By hypothesis the arrow $D\longrightarrow E$ exists. If we let $X=C\sqcup E$,
then by (\ref{disj01}) we have ${C\longrightarrow X}$ and 
${E\longrightarrow X}$, since ${D\longrightarrow E}$, by definition we have 
${D\longrightarrow X}$ and then by (\ref{disj02}), 
${C\sqcup D\longrightarrow X}$.

Let $X = C\sqcap D$, then by hypothesis and (\ref{conj01}), we have 
${X\longrightarrow C}$ and ${X\longrightarrow D\longrightarrow E}$. Using 
(\ref{conj1}), we get $X\longrightarrow C\sqcap E$.
\end{proof}

 Note that under the set-theoretical semantics  the distributive property of disjunction over conjunction is not \emph{independent}, i.e. it is a consequence of the definitions of disjunction and conjunction. However, this does not hold under the category-theoretical semantics.  The following lemma provides the connection between the usual set semantics of conjunction and the category-theoretical one given in Definition~\ref{def:conj}. In this lemma, Properties~(\ref{conj001}-\ref{conj003}) 
  are rewritings of Properties~(\ref{conj01}-\ref{conj2}) in set theory.  

\begin{lemma}\label{lem:conj}
The category-theoretical semantics of $C\sqcap D$ characterized by  Definition~\ref{def:conj}
is compatible with the set-theoretical semantics of $C\sqcap D$, that means
if $\langle \Delta^\mathcal{I}, \cdot^\mathcal{I}\rangle$ is an interpretation under set-theoretical semantics, then the following holds:  

$(U\sqcap V)^\mathcal{I}=U^\mathcal{I}\cap V^\mathcal{I}$ for all concepts $U,V$  iff
\begin{align}
&(C \sqcap D)^\mathcal{I} \subseteq  C^\mathcal{I},    (C \sqcap D)^\mathcal{I} \subseteq  D^\mathcal{I}   \label{conj001} \\
& \forall X\subseteq \Delta^\mathcal{I}, X \subseteq  C^\mathcal{I}, X  \subseteq  D^\mathcal{I} \Longrightarrow 
X  \subseteq (C \sqcap D)^\mathcal{I}  \label{conj002}\\
&(C \sqcap (D \sqcup E))^\mathcal{I} \subseteq  ((C \sqcap D)\sqcup (C \sqcap E))^\mathcal{I}  \label{conj003}\\
&\hspace{1cm}\text{for all concepts } C,D \text{ and } E.\nonumber
\end{align}
\end{lemma}

\begin{proof}
 \noindent ``$\Longleftarrow$". Due to (\ref{conj001}) we have $(U\sqcap V)^\mathcal{I}  \subseteq U^\mathcal{I}\cap V^\mathcal{I}$. Let $X\subseteq \Delta^\mathcal{I}$   such that $X =U^\mathcal{I}\cap V^\mathcal{I}$. This implies that $X \subseteq U^\mathcal{I}$ and $X \subseteq V^\mathcal{I}$.   From (\ref{conj002}), we have  $X\subseteq (U\sqcap V)^\mathcal{I}$.  
 
 \noindent ``$\Longrightarrow$".  From $(C\sqcap D)^\mathcal{I}=C^\mathcal{I}\cap D^\mathcal{I}$   we obtain (\ref{conj001}).    Moreover, if   $X\subseteq C^\mathcal{I}$ and $X\subseteq D^\mathcal{I}$ then $X\subseteq C^\mathcal{I}\cap D^\mathcal{I}=(C\sqcap D)^\mathcal{I}$ by the hypothesis. Thus, (\ref{conj002}) is proved. To prove (\ref{conj003}), we use the hypothesis and the usual set-theoretical semantics as follows:  $(C  \sqcap (D\sqcup E))^\mathcal{I}=C^\mathcal{I} \cap (D\sqcup E)^\mathcal{I}=C^\mathcal{I} \cap (D^\mathcal{I} \cup E^\mathcal{I})=(C^\mathcal{I} \cap D^\mathcal{I}) \cup (C^\mathcal{I} \cap E^\mathcal{I})=(C \sqcap D)^\mathcal{I} \cup (C \sqcap E)^\mathcal{I}=((C \sqcap D)  \sqcup (C \sqcap E))^\mathcal{I}$. 
\end{proof}

 With disjunction and conjunction, one could use the arrows $C\sqcap \neg C\longrightarrow \bot$ and $\top \longrightarrow C\sqcup \neg C$ to define  category-theoretical semantics of  negation. However, such a definition does not allow to entail useful properties (cf. Lemma~\ref{lem:neg-prop}) which are available    under set-theoretical semantics.   Therefore, it is required to use more properties to characterize negation in category-theoretical semantics. 

\begin{definition}[Category-theoretical semantics of  negation]\label{def:neg}

Let $C, \neg C, C \sqcap \neg C , C \sqcup \neg C$ be objects  of $\mathscr{C}_c\langle C, \mathcal{O}\rangle$.  Category-theoretical semantics of $\neg C$ is defined by using the following arrows in $\mathscr{C}_c\langle C, \mathcal{O}\rangle$. 
\begin{align}
& C \sqcap \neg C \longrightarrow \bot \label{neg-bot} \\
& \top\longrightarrow C \sqcup \neg C\label{neg-top}  \\
& \forall X, C \sqcap X \longrightarrow \bot \Longrightarrow X \longrightarrow \neg C \label{neg-max} \\
& \forall X, \top \longrightarrow C \sqcup X  \Longrightarrow \neg C \longrightarrow X \label{neg-min}
\end{align}
\end{definition}
  
Informally, Property~(\ref{neg-max}) tells us that $\neg C$ is the ``greatest" object satisfying (\ref{neg-bot}) while Property~(\ref{neg-min}) tells us that $\neg C$ is the ``smallest" object satisfying (\ref{neg-top}). We now provide the connection between set-theoretical semantics of negation and the category-theoretical one given in Definition~\ref{def:neg}. 

\begin{lemma}\label{lem:neg}
The category-theoretical semantics of $\neg C$ characterized by Definition~\ref{def:neg} is compatible with the set-theoretical semantics of $\neg C$, that means if $\langle \Delta^\mathcal{I}, \cdot^\mathcal{I}\rangle$ is an interpretation      under set-theoretical semantics, then the following holds:  
 
  $C^\mathcal{I} \cap \neg C^\mathcal{I}\subseteq \bot^\mathcal{I}$ and $\top^\mathcal{I} \subseteq  C^\mathcal{I} \cup \neg C^\mathcal{I}$  imply
\begin{align}
%&(C \sqcap \neg C)^\mathcal{I} \subseteq   \bot^\mathcal{I}   \label{neg1} \\
%& \top^\mathcal{I} \subseteq  (C \sqcup \neg C)^\mathcal{I} \label{neg2}\\   
& \forall X, (C \sqcap X)^\mathcal{I} \subseteq \bot^\mathcal{I} \Longrightarrow X^\mathcal{I} \subseteq (\neg C)^\mathcal{I} \label{neg3} \\
& \forall X, \top^\mathcal{I} \subseteq (C \sqcup X)^\mathcal{I}  \Longrightarrow (\neg C)^\mathcal{I} \subseteq X^\mathcal{I} \label{neg4}
\end{align}

\end{lemma}
%Note that the properties  $C^\mathcal{I} \cap \neg C^\mathcal{I}\subseteq \bot^\mathcal{I}$ and $\top^\mathcal{I} \subseteq  C^\mathcal{I} \cup \neg C^\mathcal{I}$ suffice to characterize the semantics of negation  under the set-theoretical semantics.

\begin{proof}
 
\noindent %``$\Longrightarrow$". 
%It suffices to show (\ref{neg3}) and (\ref{neg4}). 
Let $x\in X^\mathcal{I}$ with $(C\sqcap X)^\mathcal{I}\subseteq \bot^\mathcal{I}$.  Due to Lemma~\ref{lem:conj}, we have $(C\sqcap X)^\mathcal{I}=C^\mathcal{I} \cap X^\mathcal{I}$. Therefore, $x\notin C^\mathcal{I}$, and thus $x\in (\neg C)^\mathcal{I}$. Let $x\in (\neg C)^\mathcal{I}$ with $\top^\mathcal{I} \subseteq (C\sqcup X)^\mathcal{I} $. It follows that $x\notin C^\mathcal{I}$. Due to Lemma~\ref{lem:disj}, we have $(C\sqcup X)^\mathcal{I}=C^\mathcal{I} \cup X^\mathcal{I}$. Therefore,   $x\in X^\mathcal{I}$. 
\end{proof}

From Properties~(\ref{neg-bot}-\ref{neg-min}), we obtain De Morgan's laws and other properties for the category-theoretical semantics as follows.

\begin{lemma}\label{lem:neg-prop} The following properties hold.
\begin{align}
&C \longleftrightarrows  \neg \neg C \label{neg-double}\\ 
&C\longrightarrow \neg D \Longleftrightarrow   D\longrightarrow \neg C \label{neg-dual}\\
&C \sqcap D \longrightarrow \bot \Longleftrightarrow C\longrightarrow \neg D \label{neg-disjoint}\\
& \top\longrightarrow C\sqcup D\Longleftrightarrow\neg C\longrightarrow D\label{neg-conjoint}\\
&\neg (C \sqcap D) \longleftrightarrows   \neg C \sqcup  \neg D \label{neg-conj}\\
&\neg (C \sqcup D) \longleftrightarrows   \neg C \sqcap  \neg D \label{neg-disj}
\end{align}
\end{lemma}

\begin{proof}
\begin{enumerate}[wide, labelwidth=!, labelindent=0pt]
\item By Properties~(\ref{neg-bot}) and (\ref{neg-max})  where $C$ gets $\neg C$ and $X$ gets $C$, we have $C\longrightarrow \neg \neg C$. Analogously,  by Properties~(\ref{neg-top}) and (\ref{neg-min}), we obtain  $\neg \neg C \longrightarrow C$.  Hence, (\ref{neg-double}) is proved. 

\item To prove (\ref{neg-dual}), we need $C\sqcap D\longrightarrow \bot$ which follows from $C\sqcap D \longrightarrow C\longrightarrow \neg D$ and $C\sqcap D \longrightarrow D$, and thus   $C\sqcap D \longrightarrow D\sqcap \neg D\longrightarrow \bot$. Conversely, if $D\longrightarrow \neg C$, then we have ${C\sqcap D\longrightarrow D\longrightarrow \neg C}$ and ${C\sqcap D\longrightarrow C}$ thus ${C\sqcap D\longrightarrow\bot}$ and we have finished the proof.

\item To prove (\ref{neg-disjoint}), we start by proving    $C\longrightarrow \neg D  \Longrightarrow C \sqcap D \longrightarrow \bot$. We have $C\sqcap D \longrightarrow D$ and $C\sqcap D \longrightarrow C \longrightarrow \neg D$. By definition of conjunction, we obtain $C\sqcap D \longrightarrow D \sqcap \neg D\longrightarrow \bot$. To prove the other direction, we use (\ref{neg-max}) with $X=D$ or $X=C$.

\item The first direction of (\ref{neg-conjoint}) is a direct consequence of 
Definition~\ref{def:neg}, Property~(\ref{neg-min}) with $X = D$. To prove the 
other direction, we start by using (\ref{neg-top}), which gives us 
$\top\longrightarrow E\sqcup\neg E$. Since there is an arrow $E\longrightarrow F$, 
according to Property~(\ref{trans01}) of Lemma~\ref{lem:trans} there is also an
arrow $\top\longrightarrow F\sqcup\neg E$ - Property~(\ref{com01}) of 
Lemma~\ref{lem:comas} implies the existence of an arrow $\top\longrightarrow\neg 
E\sqcup F$.

\item To prove (\ref{neg-conj}), we need the definitions of conjunction, disjunction and negation. We have $C\sqcap D\longrightarrow C$ and $C\sqcap D\longrightarrow D$. Due to (\ref{neg-dual}), we obtain $\neg C \longrightarrow \neg(C\sqcap D)$ and $\neg D \longrightarrow \neg(C\sqcap D)$. By the definition of disjunction, we obtain $\neg C \sqcup \neg D \longrightarrow \neg(C \sqcap D)$. To prove the inverse, we take arrows $\neg C \longrightarrow \neg C  \sqcup  \neg D$ and $\neg D \longrightarrow \neg C  \sqcup  \neg D$ from the definition of disjunction. Due to (\ref{neg-dual}), it follows that $\neg (\neg C  \sqcup  \neg D) \longrightarrow C$ and $\neg (\neg C  \sqcup  \neg D) \longrightarrow D$. By definition, we obtain $\neg(\neg C \sqcup \neg D) \longrightarrow  C \sqcap D$. Due to (\ref{neg-dual}) and (\ref{neg-double}), it follows that $\neg(C \sqcap D) \longrightarrow  \neg C \sqcup \neg D$.

\item Analogously, we can prove (\ref{neg-disj}) by starting with arrows  $\neg C  \sqcap  \neg D  \longrightarrow \neg C$ and $\neg C  \sqcap  \neg D \longrightarrow  \neg D$ from the definition of conjunction. Due to (\ref{neg-dual}), we have $C\longrightarrow \neg(\neg C  \sqcap  \neg D)$ and $D \longrightarrow  \neg(\neg C  \sqcap  \neg D)$. By the definition of disjunction, we have $C \sqcup D \longrightarrow \neg(\neg C  \sqcap  \neg D)$, and by (\ref{neg-dual}) we obtain  $(\neg C  \sqcap  \neg D) \longrightarrow \neg(C \sqcup D)$. To prove the inverse, we take arrows $C \longrightarrow C  \sqcup  D$ and $D \longrightarrow C  \sqcup  D$ obtained from the definition of disjunction. Due to  (\ref{neg-dual}), we have $\neg(C  \sqcup  D) \longrightarrow \neg C$ and $\neg(C  \sqcup  D) \longrightarrow  \neg D$. By definition, we have $\neg(C  \sqcup  D) \longrightarrow \neg C \sqcap \neg D$.  
\end{enumerate}
\end{proof}

In order to define category-theoretical semantics of existential restrictions, we need to introduce new objects and arrows  to $\mathscr{C}_c\langle C, \mathcal{O}\rangle$ and  $\mathscr{C}_r\langle C, \mathcal{O}\rangle$,  and use the functors $\mathsf{dom}$ and  $\mathsf{cod}$  as described in Figure~\ref{fig:exist}.

\begin{figure}[!htbp]
\centering
\begin{tikzpicture}[line width=1pt,scale=1] 
\node[] (C) at (0,-1) [label={[xshift=-0.3cm, yshift=-0.6cm]}] {$C$};
\node[] (codRP) at (-2,-1) [label={[xshift=-0.3cm, yshift=-0.6cm]}] {$\mathsf{cod}(R')$};
\node[] (codRE) at (2,-1) [label={[xshift=-0.3cm, yshift=-0.6cm]}] {$\mathsf{cod}(R_{(\exists R.C)})$};

\node[] (domRP) at (-2,-2.5) [label={[xshift=-0.3cm, yshift=-0.6cm]}] {$\mathsf{dom}(R')$};
\node[] (domRE) at (2,-2.5) [label={[xshift=-0.3cm, yshift=-0.6cm]}] {$\mathsf{dom}(R_{(\exists R.C)})$};

\node[] (RP) at (-2,1) [label={[xshift=-0.3cm, yshift=-0.6cm]}] {$R'$};
\node[] (RE) at (2,1) [label={[xshift=-0.3cm, yshift=-0.6cm]}] {$R_{(\exists R.C)}$};

\node[] (R) at (0,1) [label={[xshift=-0.3cm, yshift=-0.6cm]}] {$R$};

 \node[] (X) at (0,1) [label={[xshift=-0.3cm, yshift=-0.6cm]}] { };

 \draw[->] (codRP) to[]node[below=2pt]{${i}$} (C);
 
 \draw[->] (codRE) to[]node[below=2pt]{${j}$} (C);
 \draw[dashed,->] (RP) to[]node[right=2pt]{$\mathsf{cod}$} (codRP);
 
 \draw[dashed,->] (RE) to[]node[left=2pt]{$\mathsf{cod} $} (codRE);
 
 \draw[->] (RP) to[]node[above=2pt]{${k}$} (X); 
 \draw[->] (RE) to[]node[above=2pt]{${l}$} (X); 
  
 \path [->,dotted,out=210,in=135] (RP) edge node[left]{$\mathsf{dom}$} (domRP);
  \path [->,dotted,out=-30,in=45] (RE) edge node[right]{$\mathsf{dom}$} (domRE);
 \draw[->] (domRP) to[]node[above=2pt]{$m$} (domRE);
\end{tikzpicture}
\caption{Category-theoretical semantics for existential restriction}\label{fig:exist}
\end{figure}
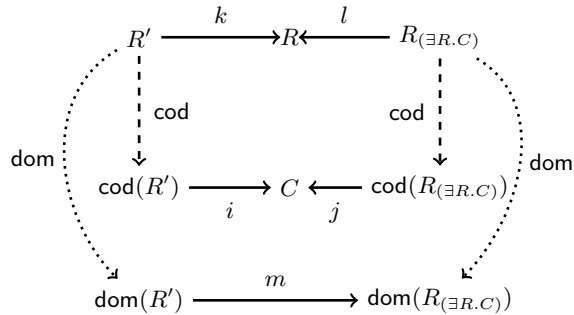

\begin{definition}[Category-theoretical semantics of  existential restriction]\label{def:exist}
Let $\exists R.C, C$ be objects of $\mathscr{C}_c\langle C, \mathcal{O}\rangle$, and    
$R, R_{(\exists R.C)}$ be objects of  $\mathscr{C}_r\langle C, \mathcal{O}\rangle$.
%$l$ be an arrow from $R_{(\exists R.C)}$ to $R$,  and $\mathsf{dom}(R_{(\exists R.C)})$, $\mathsf{cod}(R_{(\exists R.C)})$ be the images of $R_{(\exists R.C)}$ in  $\mathscr{C}_c\langle C, \mathcal{O}\rangle$ by the functors $\mathsf{dom}$ and $\mathsf{cod}$.   
Category-theoretical semantics of $\exists R.C$ is defined by using arrows in $\mathscr{C}_c\langle C, \mathcal{O}\rangle$ and $\mathscr{C}_r\langle C, \mathcal{O}\rangle$ as follows.  
\begin{align}
&R_{(\exists R.C)}  \longrightarrow R, \mathsf{cod}(R_{(\exists R.C)})  \longrightarrow C\label{ex-arrow01}\\ &\mathsf{dom}(R_{(\exists R.C)})\longleftrightarrows \exists R.C \label{ex-arrow02}\\
&\forall R', R'\longrightarrow R, \mathsf{cod}(R') \longrightarrow C\Longrightarrow \label{ex-arrow03} \\ &\hspace{2cm}\mathsf{dom}(R')\longrightarrow \mathsf{dom}(R_{(\exists R.C)})  \nonumber
\end{align}
\end{definition}
 The following lemma  establishes the connection between set-theoretical semantics of existential restriction and the category-theoretical one given in Definition~\ref{def:exist}. 

\begin{lemma}\label{lem:exist}
The category-theoretical semantics of $\exists R.C$ characterized by Definition~\ref{def:exist} is compatible with the set-theoretical semantics of $\exists R.C$, that means
if $\langle \Delta^\mathcal{I}, \cdot^\mathcal{I}\rangle$ is an interpretation      under set-theoretical semantics such that  
\begin{align}
&R_{(\exists R.C)}^\mathcal{I} \subseteq R^\mathcal{I} \label{ex-arrow1}\\
&\mathsf{cod}(R_{(\exists R.C)})^\mathcal{I} \subseteq C^\mathcal{I} \label{ex-arrow2}
\end{align}
\noindent then the following holds:

\noindent  $\mathsf{dom}(R_{(\exists R.C)})^\mathcal{I}=\{x\in \Delta^\mathcal{I}\mid \exists y\in \Delta^\mathcal{I} : \langle x,y\rangle\in R^\mathcal{I} \wedge y\in C^\mathcal{I}\}$ iff
\begin{align}
&\forall R'\subseteq \Delta^\mathcal{I}\times \Delta^\mathcal{I}, {R'}\subseteq R^\mathcal{I}, \mathsf{cod}(R') \subseteq C^\mathcal{I}\Longrightarrow \label{ex-arrow3}\\ &\hspace{3cm}\mathsf{dom}(R')\subseteq \mathsf{dom}(R_{(\exists R.C)})^\mathcal{I}   \nonumber
\end{align}
\end{lemma}

\begin{proof}
 \noindent ``$\Longleftarrow$". 
Let $x'\in \mathsf{dom}(R_{(\exists R.C)})^\mathcal{I}$.   There is an element $y\in   \mathsf{cod}(R_{(\exists R.C)})^\mathcal{I}$ such that $\langle x',y\rangle \in R_{(\exists R.C)}^\mathcal{I}$ by definition. Due to (\ref{ex-arrow2}),   we have $y\in   C^\mathcal{I}$. Analogously, due to (\ref{ex-arrow1}),     we have $\langle x',y\rangle \in R^\mathcal{I}$. Thus, $x'\in  \{x\in \Delta^\mathcal{I}\mid \exists y\in \Delta^\mathcal{I} : \langle x,y\rangle\in R^\mathcal{I} \wedge y\in C^\mathcal{I}\}$.

  Let $x'\in  \{x\in \Delta^\mathcal{I}\mid \exists y\in \Delta^\mathcal{I} : \langle x,y\rangle\in R^\mathcal{I} \wedge y\in C^\mathcal{I}\}$.  Thus, there is an element $y\in   C^\mathcal{I}$ such that $\langle x',y\rangle \in R^\mathcal{I}$.  Take an $R'\subseteq \Delta^\mathcal{I} \times \Delta^\mathcal{I}$  such that $R'\subseteq R^\mathcal{I}$,  $\mathsf{cod}(R') \subseteq C^\mathcal{I}$ with   $x'\in \mathsf{dom}(R')$ and  $y\in \mathsf{cod}(R')$.   Due to (\ref{ex-arrow3}), we have $x'\in \mathsf{dom}(R_{(\exists R.C)})^\mathcal{I}$.
  
  \noindent ``$\Longrightarrow$". Let $R'\subseteq \Delta^\mathcal{I} \times \Delta^\mathcal{I}$   such that $R'\subseteq R^\mathcal{I}$, $\mathsf{cod}(R') \subseteq C^\mathcal{I}$. Let $x'\in \mathsf{dom}(R')$. This implies that there is some $y\in \mathsf{cod}(R')$ such that $\langle x',y\rangle\in  R'$. Hence, $y\in C^\mathcal{I}$ and $\langle x',y\rangle\in  R^\mathcal{I}$. Therefore, $x'\in  \{x\in \Delta^\mathcal{I}\mid \exists y\in \Delta^\mathcal{I} : \langle x,y\rangle\in R^\mathcal{I} \wedge y\in C^\mathcal{I}\}$ $=$ $\mathsf{dom}(R_{(\exists R.C)})^\mathcal{I}$. 
\end{proof}

\begin{lemma}\label{lem:exist-prop} The following properties hold:
\begin{align}
&C \longrightarrow \bot \Longrightarrow \exists R.C  \longrightarrow \bot \label{exist-empty}\\
&C \longrightarrow D \Longrightarrow \exists R.C  \longrightarrow \exists R.D  \label{exist-sub}
%\\
%& \exists R.(C\sqcup D) \longleftrightarrows  \exists R.C \sqcup \exists R.D \label{exist-disj}
\end{align}
\end{lemma}
\begin{proof}
\begin{enumerate}
    \item We have $\mathsf{cod}(R_{(\exists R.C)})\longrightarrow C$ by definition. By hypothesis $C\longrightarrow \bot$, we have $\mathsf{cod}(R_{(\exists R.C)})\longrightarrow \bot$ due Definition~\ref{def:syntax-cat} of $\mathsf{cod}$. Again,  due Definition~\ref{def:syntax-cat} of $\mathsf{cod}$, we have  $R_{(\exists R.C)} \longrightarrow R_\bot$. Moreover, we have $\mathsf{dom}(R_{(\exists R.C)})  \longrightarrow \bot$  due Definition~\ref{def:syntax-cat} of $\mathsf{dom}$. By Definition~\ref{def:exist}, we obtain $\exists R.C\longrightarrow \bot$.
    
    \item To prove (\ref{exist-sub}) we consider two objects $R_{(\exists R.C)}$ and $R_{(\exists R.D)}$ in $\mathscr{C}_r$. We have $R_{(\exists R.C)}\longrightarrow R$ and  $\mathsf{cod}(R_{(\exists R.C)})\longrightarrow C \longrightarrow D$. By definition, we obtain $\mathsf{dom}(R_{(\exists R.C)})\longrightarrow \mathsf{dom}(R_{(\exists R.D)})$.
    \end{enumerate}
\end{proof} 
 
The set-theoretical semantics of universal restriction can be defined as  negation of existential restriction. However,  $\forall R.C \longleftrightarrows  \neg  \exists R.\neg C$ does not allow to obtain  usual connections between existential and universal restrictions such as $\exists R.D \sqcap \forall R.C \longrightarrow  \exists R.(D \sqcap  C)$  under the category-theoretical semantics (the ``property (ii)'' mentioned in our introduction).   Therefore, we need more arrows to define category-theoretical semantics of universal restriction as follows.

\begin{definition}[Category-theoretical semantics of universal restriction]\label{def:forall}
Let $\forall R.C$, $C$  be objects of $\mathscr{C}_c\langle C, \mathcal{O}\rangle$,  and    
$R$ be an object of  $\mathscr{C}_r\langle C, \mathcal{O}\rangle$.
   Category-theoretical semantics of $\forall R.C$ is defined by using arrows in $\mathscr{C}_c\langle C, \mathcal{O}\rangle$ and $\mathscr{C}_r\langle C, \mathcal{O}\rangle$ as follows.
\begin{align}
&\forall R.C \longleftrightarrows  \neg  \exists R.\neg C\label{all-arrow1} \\
&\forall R', R'\longrightarrow R, \mathsf{dom}(R') \longrightarrow \forall R.C \Longrightarrow \mathsf{cod}(R')\longrightarrow C \label{all-arrow2}  
\end{align}
\end{definition}

We formulate and prove the connection between the usual set-theoretical semantics of universal restriction and the category-theoretical one given in Definition~\ref{def:forall}.

\begin{lemma}\label{lem:forall}
The category-theoretical semantics of $\forall R.C$ characterized by Definition~\ref{def:forall}  is compatible with the set-theoretical semantics of $\forall R.C$, that means
if $\langle \Delta^\mathcal{I}, \cdot^\mathcal{I}\rangle$   is an interpretation   under set-theoretical semantics, then the following holds:   
  
  $\forall R.C^\mathcal{I}=\{x\in \Delta^\mathcal{I}\mid  \langle x,y\rangle\in R^\mathcal{I}\Longrightarrow y\in C^\mathcal{I}\}$ iff
\begin{align}
&\forall R.C^\mathcal{I}= (\neg  \exists R.\neg C)^\mathcal{I} \label{all-arrow4}\\
&\forall R'\subseteq \Delta^\mathcal{I}\times \Delta^\mathcal{I}, R' \subseteq R^\mathcal{I}, \mathsf{dom}(R')  \subseteq  \forall R.C^\mathcal{I}   \Longrightarrow  \label{all-arrow5}\\  &\hspace{4cm}\mathsf{cod}(R')^\mathcal{I}\subseteq C^\mathcal{I}  \nonumber
\end{align}
\end{lemma}
\begin{proof}
 \noindent ``$\Longleftarrow$". We have $\{x\in \Delta^\mathcal{I}\mid  \langle x,y\rangle\in R^\mathcal{I}\Longrightarrow y\in C^\mathcal{I}\}$ is the complement of  $\{x\in \Delta^\mathcal{I}\mid \exists y\in \Delta^\mathcal{I} : \langle x,y\rangle\in R^\mathcal{I}\wedge y\in \neg C^\mathcal{I}\} = \exists R.\neg C^\mathcal{I}$. Hence, we have  $\forall R.C^\mathcal{I}=\{x\in \Delta^\mathcal{I}\mid  \langle x,y\rangle\in R^\mathcal{I}\Longrightarrow y\in C^\mathcal{I}\}$ from (\ref{all-arrow4}).
 
 \noindent ``$\Longrightarrow$". Let  $R'\subseteq R^\mathcal{I}$ with $\mathsf{dom}(R') \subseteq \forall R.C^\mathcal{I}$. Let $y\in \mathsf{cod}(R')$. There is some $x'\in \mathsf{dom}(R')$ such that $\langle x', y\rangle \in  {R'}^\mathcal{I}$. Since $R'\subseteq R^\mathcal{I}$ and  $\mathsf{dom}(R') \subseteq \forall R.C^\mathcal{I}$, we have $x\in \forall R.C^\mathcal{I}$ and $\langle x', y\rangle \in  {R}^\mathcal{I}$. By hypothesis, $y\in C^\mathcal{I}$. 
\end{proof}

\begin{lemma}\label{lem:forall-prop} The following properties hold. 
\begin{align}
&\forall R.C \sqcap \exists  R.\neg C \longrightarrow  \bot \label{forall-exists-bot}\\
&C \longrightarrow D \Longrightarrow \forall R.C \longrightarrow  \forall R.D  \label{forall-sub}\\
&\exists R.D \sqcap \forall R.C \longrightarrow  \exists R.(D \sqcap  C) \label{forall-exist}\\
&\exists R.C \sqcup \exists R.D  \longleftrightarrows  \exists R.(C \sqcup D) \label{exists-disj}
\end{align}
\end{lemma}

\begin{proof}
\begin{enumerate}[wide, labelwidth=!, labelindent=0pt]
\item  (\ref{forall-exists-bot}) is a consequence of (\ref{all-arrow1}) and (\ref{neg-bot}).

    \item Due to (\ref{neg-dual}), $C\longrightarrow D$ implies $\neg D \longrightarrow \neg C$. By (\ref{exist-sub}), we obtain $\exists R.\neg D\longrightarrow \exists R.\neg C$. Again, due to (\ref{neg-dual}), we have $\neg \exists R.\neg C\longrightarrow \neg \exists R.\neg D$. By definition with (\ref{all-arrow1}), it follows $\forall R.C \longrightarrow \forall R.D$.

    \item To prove (\ref{forall-exist}) we need category-theoretical semantics of existential and universal restrictions. We define an object $R_X$ in $\mathscr{C}_r$ such that $R_X\longrightarrow R_{(\exists R.D)}$ and  $\mathsf{dom}(R_X)\longleftrightarrows \exists R.D \sqcap \forall R.C$. This implies that $\mathsf{cod}(R_X)\longrightarrow  \mathsf{cod}(R_{(\exists R.D)})$ due to the functor $\mathsf{cod}$ from $\mathscr{C}_r$ to $\mathscr{C}_c$, and  $\mathsf{dom}(R_X)\longrightarrow   \forall R.C$. By definition, we have $\mathsf{cod}(R_{(\exists R.D)})\longrightarrow D$. Thus, $\mathsf{cod}(R_X)\longrightarrow D$.    Due to Arrow~(\ref{all-arrow2}) and $\mathsf{dom}(R_X)\longrightarrow   \forall R.C$, we obtain $\mathsf{cod}(R_X)\longrightarrow  C$. Hence, $\mathsf{cod}(R_X)\longrightarrow  C\sqcap D$. By the definition of existential restrictions, we obtain $\mathsf{dom}(R_X)\longrightarrow   \exists  R.(D \sqcap C)$. 
    
    \item The arrow from left to right of (\ref{exists-disj}) 
    is a direct consequence of the 
    the disjunction arrows~(\ref{disj01}), (\ref{disj02}) and the existential 
    property arrow~(\ref{exist-sub}). To prove the other direction, we make extensive use of 
    the negation and its properties. Let us write $U$ the right part of the arrow
    and $V$ the left part, we need to prove that the arrow $U\sqcap\neg 
    V\longrightarrow\bot$ exists. Indeed, if $U\sqcap\neg V\longrightarrow\bot$ holds, it follows that $U\longrightarrow
    \neg\neg V$ by the implication (\ref{neg-max}) in Definition \ref{def:neg}. Thus, we obtain $U\longrightarrow  V$ due to Property~(\ref{neg-double}) in Lemma \ref{lem:neg-prop}.
    
    For that, we need to rewrite $V$. First, apply rule~(\ref{neg-double}), 
    then apply our version of De Morgan's rules~(\ref{neg-disj}). Apply~(\ref{all-arrow1}) to each member of the conjunction, then (\ref{neg-double}) again:
    \[
    \exists R.C\sqcup\exists R.D\longleftrightarrows\neg(\forall R.\neg 
    C\sqcap\forall R.\neg D) 
    \]
    
    From there, consider the concept object $U\sqcap\neg V$ defined by 
    \[
    \exists R.(C\sqcap D)\sqcap\left(\forall R.\neg C\sqcap\forall R.\neg D\right)
    \]
    Using the implied associative nature of $\sqcap$, and Property~(\ref{forall-exist}) proved above, we can write
    \begin{align*}
        \exists R.(C\sqcup D)\sqcap&(\forall R.\neg C\sqcap\forall R.\neg D)\\
        &\longrightarrow\exists R.\left(\big((C\sqcup D)\sqcap\neg C\big)\sqcap\neg D\right)
    \end{align*}
    The distributive property~(\ref{conj2}) of $\sqcap$ and $\sqcup$, the disjunction property~(\ref{disj02}), and Definition~\ref{def:syntax-cat} of $\bot$, allow us to make the following reduction
    \begin{align*}
        (C\sqcup D)\sqcap\neg C&\longrightarrow (C\sqcap\neg C)\sqcup(D\sqcap\neg C)\\
        &\longrightarrow\bot\sqcup(D\sqcap\neg C)\\
        &\longrightarrow D\sqcap\neg C
    \end{align*}
    Thus, thanks to the commutative nature of $\sqcap$ and (\ref{conj01}) again, we have 
    $(D\sqcap\neg C\sqcap\neg D)\longrightarrow(\bot\sqcap C)\longrightarrow\bot$.
    Applying Properties~(\ref{exist-empty}) and (\ref{exist-sub}) of Lemma \ref{lem:exist-prop}
    to the previous existential arrow, we obtain 
    \[
    \exists R.(C\sqcup D)\sqcap(\forall R.\neg C\sqcap\forall R.\neg D)\longrightarrow\bot
    \]
    which concludes the proof.
\end{enumerate}
\end{proof}

Note that both $\mathscr{C}_c\langle C, \mathcal{O}\rangle$ and $\mathscr{C}_r\langle C, \mathcal{O}\rangle$ may consist of more  objects. However, new arrows should be derived from those existing or by using the properties given in Definitions~(\ref{def:syntax-cat}-\ref{def:forall}). Adding to $\mathscr{C}_c\langle C, \mathcal{O}\rangle$ a new arrow that is \emph{independent} from those existing leads to a semantic change of the ontology.   Since all properties in Lemma~ \ref{lem:neg-prop}, \ref{lem:exist-prop} and \ref{lem:forall-prop} are consequences of those given in these definitions, they can be used to obtain \emph{derived}  arrows (i.e. not independent ones).

\begin{theorem}[Arrow and subsumption]\label{thm:arrow-entail} Let $C_0$ be an $\mathcal{ALC}$ concept and   $\mathcal{O}$ an $\mathcal{ALC}$ ontology. Let  $\mathscr{C}_c\langle C_0,\mathcal{O} \rangle$ be an ontology  category. It holds that $\langle C_0,\mathcal{O}\rangle \models X\sqsubseteq Y$ (under set-theoretical semantics) if $X\longrightarrow Y$ is an arrow in $\mathscr{C}_c\langle C_0, \mathcal{O}\rangle$.
\end{theorem}

\begin{proof} 
For every axiom $E\sqsubseteq F$ in $\mathcal{O}$, we have $E\longrightarrow F$
in $\Catc\tuple{C_0,\mathcal{O}}$ according to Definition~\ref{def:onto-cat}. 
Property~(\ref{neg-conjoint}) of Lemma~\ref{lem:neg-prop} and Lemma~\ref{lem:neg}, 
we can conclude that $\tuple{C_0,\mathcal{O}}\models E\sqsubseteq F$.
For every other arrows $X\longrightarrow Y$ introduced to $\Catc\tuple{C_0, 
\mathcal{O}}$ from Definitions~\ref{def:disj}, \ref{def:conj}, \ref{def:neg}, 
\ref{def:exist},  \ref{def:forall}.(\ref{all-arrow1}), we have 
Lemmas~\ref{lem:disj}, \ref{lem:conj}, \ref{lem:neg}, \ref{lem:exist} and 
\ref{lem:forall}.(\ref{all-arrow4}) which give us $\langle C_0,\mathcal{O}\rangle  
\models X\sqsubseteq Y$.$\hfill\square$
\end{proof}

 We now  introduce  category-theoretical satisfiability of an  $\mathcal{ALC}$ concept  with respect to an $\mathcal{ALC}$  ontology.

\begin{definition}\label{def:cat-satisfiability}
Let $C_0$ be an $\mathcal{ALC}$ concept, $\mathcal{O}$ an $\mathcal{ALC}$ ontology.  $C$ is category-theoretically unsatifiable with respect to $\mathcal{O}$ if  there is an ontology category $\mathscr{C}_c\langle C_0, \mathcal{O}\rangle$ which has an arrow $C_0\longrightarrow \bot$.
\end{definition}

Since an ontology category $\mathscr{C}_c\langle C_0, \mathcal{O}\rangle$  may consist of objects arbitrarily built from the signature, Definition~\ref{def:cat-satisfiability} offers possibilities to build a larger ontology category  from which  new arrows can be discovered by applying arrows given in Definitions~(\ref{def:disj}-\ref{def:forall}).

\begin{theorem}\label{thm:cat-set} Let $\mathcal{O}$  be an $\mathcal{ALC}$ ontology and $C_0$ an  $\mathcal{ALC}$ concept. $C_0$ is category-theoreti\-cally unsatifiable with respect to $\mathcal{O}$ iff $C_0$ is set-theoretically unsatifiable.
\end{theorem}
To prove this theorem, we need the following preliminary result.

\begin{lemma}\label{lem:set2cat}
Let $\mathcal{O}$  be an $\mathcal{ALC}$ ontology and $C_0$  an  $\mathcal{ALC}$ concept. If $C_0$ is set-theoretically unsatifiable with respect to $\mathcal{O}$ then $C_0$ is  category-theoretically unsatifiable. 
\end{lemma}

In order to prove this lemma, we use a tableau algorithm to generate from an unsatisfiable  $\mathcal{ALC}$ concept with respect to an ontology a set of completion trees each of which contains a clash ($\bot$ or a pair $\{A,\neg A\}$ where $A$ is a concept name).  To ensure self-containedness of the paper, we describe   here necessary elements which allow to follow the proof of the lemma. We refer the readers to \cite{baa2000,hor07} for formal details.     

We use $\mathsf{sub}(C_0,\mathcal{O})$ to denote  a set of subconcepts in NNF (i.e. negations appear only in front of concept name) from $C_0$ and $\mathcal{O}$. This set is defined as the smallest set such that  the following conditions hold :
\begin{enumerate}
    \item $C_0 \in\mathsf{sub}(C_0,\mathcal{O})$ and $\neg E \sqcup F \in\mathsf{sub}(C_0,\mathcal{O})$ for each axiom $E\sqsubseteq F$; 
    \item if $\neg C \in\mathsf{sub}(C_0,\mathcal{O})$ then $C\in \mathsf{sub}(C_0,\mathcal{O})$; 
    \item if $C \sqcap D$ or  $C \sqcup D\in\mathsf{sub}(C_0,\mathcal{O})$ then $C, D\in \mathsf{sub}(C_0,\mathcal{O})$; \item  if $\exists R.C $ or  $\forall R.C \in\mathsf{sub}(C_0,\mathcal{O})$ then $C\in \mathsf{sub}(C_0,\mathcal{O})$. 
\end{enumerate}    
    A \emph{completion tree} $T=\langle v_0, V, E, L\rangle$ is a tree where $V$ is a set of nodes, and each node $x\in V$ is labelled with  $L(x)\subseteq \mathsf{sub}(C_0,\mathcal{O})$, and  a root $v_0\in V$ with $C\in L(v_0)$; $E$ is a set of edges, and each edge $\langle x,y\rangle\in E$ is labelled with a role $L\langle x,y\rangle=\{R\}$ and $R\in \mathbf{R}$. In a completion tree $T$, a node $x$ is \emph{blocked} by an ancestor $y$ if $L(x)=L(y)$. 

To build completion trees, a tableau algorithm starts by initializing a tree $T_0$ and  applies the following \emph{completion rules} to each \emph{clash-free} tree $T$ (i.e. no node contains a clash):
\begin{enumerate}
    \item 
\textbf{[$\sqsubseteq$-rule]} for each axiom $E\sqsubseteq F$ of $\mathcal{O}$ and each node $x$, if $\neg E\sqcup F \notin L(x)$ then  $L(x)\gets L(x) \cup \{\neg E\sqcup F\}$;
\item \textbf{[$\sqcap$-rule]} if $E \sqcap F\in L(x)$ and $\{E,F\}\nsubseteq L(x)$ then $L(x)\gets L(x) \cup \{E,F\}$;
\item \textbf{[$\sqcup$-rule]} if $E \sqcup F\in L(x)$ and $\{E,F\}\cap L(x)=\emptyset$ then it creates two copies $T_1$ and $T_2$ of the current tree $T$,  and set $L(x_1)\gets L(x_1)\cup \{E\}$, $L(x_2)\gets L(x_2)\cup \{F\}$ where $x_i$ in  $T_i$   is a copy of $x$ from $T$. 
\item \textbf{[$\exists$-rule]} if $\exists R.C\in L(x)$, $x$ is not blocked and $x$ has no  edge $\langle x,x'\rangle$ with $L(\langle x,x'\rangle)=\{R\}$ and $C\in L(x')$ then it creates a successor $x'$ of $x$, and set $L(x')\gets\{C\}$, $L(\langle x,x'\rangle\gets\{R\}$;
\item \textbf{[$\forall$-rule]} if $\forall R.C\in L(x)$ and $x$ has a successor $x'$ with    $L(\langle x,x'\rangle=\{R\}$ then  $L(x')\gets L(x')\cup \{C\}$. 
\end{enumerate}

When   \textbf{[$\sqcup$-rule]} is applied   to a node $x$ of a completion tree $T$ with $E\sqcup F\in L(x)$, it   generates two children trees $T_1$ and $T_2$  with a node $x_1$ in $T_1$ such that $E\sqcup F,E\in L(x_1)$,   and a node $x_2$ in $T_2$ such that $E\sqcup F,F\in L(x_2)$ as described above.  In this case,  we say that $T$ is parent of $T_1$ and $ T_2$ by $x$; or $T_1$ and $ T_2$ are children of $T$ by $x$;   
$x$ is called a \emph{disjunction} node by $E\sqcup F$; and $x_1,x_2$ are called  \emph{disjunct} nodes of $x$ by $E\sqcup F$.  We use $\mathbb{T}$ to denote the tree whose nodes are completion trees generated by the tableau algorithm as described above. At any moment, a complete rule is applied only to leaf trees of $\mathbb{T}$.

A completion tree is \emph{complete} if no completion rule is applicable. It was shown that if $C$ is set-theoretically unsatisfiable with respect to $\mathcal{O}$ then all \emph{complete completion trees} contain a clash (an incomplete completion tree may contain a clash)  \cite{baa2000,hor07}. Since this result does not depend on the order of applying completion rules to nodes, we can assume in this proof that the following order   \textbf{[$\sqsubseteq$-rule]}, \textbf{[$\sqcap$-rule]}, \textbf{[$\sqcup$-rule]}, \textbf{[$\exists$-rule]} and \textbf{[$\forall$-rule]} is used, and a completion rule should be applied to the most ancestor node if applicable. 
Some properties are drawn from this assumption.
\begin{enumerate}[label=(P\arabic*), leftmargin=*]
\item\label{prop:1} If a node $x$ of a completion tree $T$ contains a clash then $x$ is a leaf node of $T$.

\item\label{prop:2} For a completion tree $T$, if  a node $y$ is an ancestor of a node $x$  in  $T$ such that $x$ is a disjuntion node  and $y$ is a disjunt node of $y'$, then the children trees by $y'$ are ancestors of the children trees by $x$   in $\mathbb{T}$.
\end{enumerate}
\ref{prop:1} tells us that when a clash is discovered in a completion tree $T$, no rule is applied to any node of $T$ while \ref{prop:2} is a consequence of the order of rule applications and the fact that each completion rule is applied to the most ancestor node in a completion tree if applicable.

\noindent \textbf{Proof of Lemma~\ref{lem:set2cat}}. We define an ontology category $\mathscr{C}_c\langle C_0, \mathcal{O}\rangle$  from $\mathbb{T}$ by starting from leaf trees of   $\mathbb{T}$. By \ref{prop:1}, each leaf tree $T$ of  $\mathbb{T}$ contains a clash in a leaf node $y$. We add an object $\bigsqcap_{Y\in L(y)} Y$ and  an  arrow  $\bigsqcap_{Y\in L(y)} Y\longrightarrow \bot$  to $\mathscr{C}_c\langle C_0, \mathcal{O}\rangle$. In the sequel, we try to define a sequence of arrows started with $\bigsqcap_{Y\in L(y)} Y\longrightarrow \bot$  which makes clashes propagate from the leaves into the root of $\mathbb{T}$. This propagation has to get through two crucial kinds of passing: from a node of a completion tree to an ancestor that is a disjunct node; and from such a disjunct node in a completion tree to its disjunction node in the parent completion tree in $\mathbb{T}$.

Let $T_1$ and $T_2$ be two leaf trees of $\mathbb{T}$ whose parent is $T$, and $y_1, y_2$ two leaf nodes of $T_1,T_2$ containing a clash.  By construction, $T_i$ has  a disjunct node $x_i$  which is an ancestor of $y_i$     with  $x_i\neq y_i$, and there is no disjunct node between $x_i$ and $y_i$ (among descendants of $x_i$). For each node $z$ between $y_i$ and $x_i$ if it exists, we add objects $\bigsqcap_{Z\in L(z)} Z$   to $\mathscr{C}_c\langle C_0, \mathcal{O}\rangle$.  Let $z'$ be the parent node of $y_i$. By construction, there is  a concept $\exists R.D\in L(z')$ and $D\in L(y_i)$. In addition, if  $\forall R.D_1\in L(z')$ and \textbf{[$\forall$-rule]} is applied to $z'$ then  $D_1\in L(y_i)$. We show that $D\longrightarrow Z$ or  $D_1\longrightarrow Z$ for each concept $Z\in L(y_i)$. By construction,    \textbf{[$\sqcup$-rule]} is not applied to $y_i$ (and any node from $x_i$ to $y_i$). If $Z=\neg E \sqcup F$ comes from an axiom $E\sqsubseteq F$, then $D\longrightarrow \top \longrightarrow Z$. If $Z$ comes from  some $Z \sqcap Z'$  then we have to have already  $D\longrightarrow  Z \sqcap Z'\longrightarrow Z$ or $D_1\longrightarrow  Z \sqcap Z'\longrightarrow Z$. Hence, $\bigsqcap_{Y\in L(y_i)} Y\longrightarrow \bot$ implies $D\sqcap \bigsqcap_{\forall R.D_i\in L(z')} D_i\longrightarrow \bot$. By (\ref{exist-empty}) and  (\ref{forall-exist}), we obtain $\exists R.D\sqcap \bigsqcap_{\forall R.D_i\in L(z')} \forall R.D_i\longrightarrow \bot$, and thus $\bigsqcap_{Z\in L(z')} Z\longrightarrow \bot$. We add an object $\bigsqcap_{X\in L(z')} X$ and an arrow  $\bigsqcap_{X\in L(z')} X\longrightarrow \bot$   to $\mathscr{C}_c\langle C_0, \mathcal{O}\rangle$.  By using the same argument from $z'$ until $x_i$, we add an object $\bigsqcap_{X\in L(x_i)} X$ and an arrow  $\bigsqcap_{X\in L(x_i)} X\longrightarrow \bot$   to $\mathscr{C}_c\langle C_0, \mathcal{O}\rangle$.

We now show $\bigsqcap_{X\in L(x)} X\longrightarrow \bot$ where $x$ is the disjunction node of $x_i$, $L(x)=W\cup \{E\sqcup F\}$,  $L(x_1)=W\cup \{E\sqcup F, E\}$, $L(x_2)=W\cup \{ E\sqcup F, F\}$. Due to (\ref{conj2}), we have $\bigsqcap_{V\in W} V\sqcap (E\sqcup F) \longrightarrow (\bigsqcap_{V\in W} V\sqcap  E) \sqcup (\bigsqcap_{V\in W} V \sqcap F)$. Moreover, since $\bigsqcap_{X\in L(x_i)} X\longrightarrow \bot$, we have $ \bigsqcap_{V\in W} V\sqcap  E \longrightarrow \bigsqcap_{V\in W} V\sqcap (E\sqcup F) \sqcap E \longrightarrow  \bigsqcap_{X\in L(x_1)} X \longrightarrow\bot$ (we have $(E\sqcup F)\sqcap E\longrightarrow E$ due to (\ref{disj01}) and (\ref{conj1})), and $\bigsqcap_{V\in W} V\sqcap  F \longrightarrow \bigsqcap_{V\in W} V\sqcap (E\sqcup F) \sqcap F \longrightarrow  \bigsqcap_{X\in L(x_2)} X \longrightarrow\bot$. Therefore, if we add an object  $\bigsqcap_{V\in W} V\sqcap (E\sqcup F)$ to $\mathscr{C}_c\langle C_0, \mathcal{O}\rangle$, we obtain an arrow  $\bigsqcap_{V\in W} V\sqcap (E\sqcup F) \longrightarrow\bot$ due to (\ref{conj2}).

We can apply the same argument from $x$ to the next disjunct node  in $T$ which is an ancestor of $x$ according to \ref{prop:2},  and go upwards in $\mathbb{T}$ to find its parent and sibling. This process can continue and reach the root tree $T_0$ of $\mathbb{T}$. We get   $\bigsqcap_{X\in L(x_0)} X \longrightarrow \bot$ where $x_0$ is the root node of $T_0$. By construction, we have $\{C\}\cup \{\neg E\sqcup F\mid E\sqsubseteq F\in \mathcal{O}\}\subseteq L(x_0)$.  Let $X\in L(x_0)$. If $X=\neg E \sqcup F$ for some axiom $E  \sqsubseteq F\in \mathcal{O}$ then  $C\longrightarrow \top \longrightarrow \neg E \sqcup F$. Moreover, if $X=E'$ or $X=F'$ with $C= E' \sqcap F'$ then $C\longrightarrow E'$ and $C\longrightarrow F'$. This implies that    $C  \longrightarrow  \bigsqcap_{X\in L(x_0)} X \longrightarrow \bot$. Hence,   $C \longrightarrow  \bot$.  Therefore,  $\mathscr{C}_c\langle C_0, \mathcal{O}\rangle$ is an ontology category consisting of  $C \longrightarrow  \bot$.  This completes the proof of the lemma.

%Moreover, an arrow $\top \longrightarrow\neg E\sqcap F$ is added to $\mathscr{C}_c\langle C, \mathcal{O}\rangle$ for each axiom $E\sqsubseteq F\in \mathcal{O}$.

\noindent \textbf{Proof of Theorem~\ref{thm:cat-set}}.
\noindent ``$\Longleftarrow$". Let $\mathscr{C}_c\langle C_0, \mathcal{O}\rangle$ be an ontology category. Assume that there is   an arrow $C\longrightarrow \bot$ in $\mathscr{C}_c\langle C_0, \mathcal{O}\rangle$. Every
 arrow $X\longrightarrow Y$ added to $\mathscr{C}_c\langle C_0, \mathcal{O}\rangle$  must be one of following cases:
 %\begin{enumerate}
     (i) $X\sqsubseteq Y$ is an axiom of $\mathcal{O}$. Thus, $\mathcal{O}\models X\sqsubseteq Y$.
     (ii) $X\longrightarrow Y$ is added by Definitions~\ref{def:disj}, \ref{def:conj}, \ref{def:neg}, \ref{def:exist},  \ref{def:forall}. Due to  Theorem~\ref{thm:arrow-entail}, we have $\mathcal{O}\models X\sqsubseteq Y$.
     (iii) $X\longrightarrow Y$ is obtained by transitivity from  $X\longrightarrow Z$ and $Z\longrightarrow Y$. It holds that    $\mathcal{O}\models X\sqsubseteq Z$ and $\mathcal{O}\models Z\sqsubseteq Y$ imply $\mathcal{O}\models X\sqsubseteq Y$.
% \end{enumerate}
Hence, if $\mathscr{C}_c\langle C_0, \mathcal{O}\rangle$ consists of   an arrow $C\longrightarrow \bot$, then $\mathcal{O}\models C\sqsubseteq \bot$.

\noindent ``$\Longrightarrow$". A consequence of Lemma~\ref{lem:set2cat}.

%%%%%%%%%%%%%%%%%%%%%%%%
%
%%%%%%%%%%%%%%%%%%%%%%%%
\section{Reasoning in a sublogic of $\mathcal{ALC}$} \label{sec:reasoning}
 
 In this section, we identify a new sublogic of $\mathcal{ALC}$,  namely $\mathcal{ALC}_{\overline{\forall}}$, obtained from $\mathcal{ALC}$ with all the properties introduced in Definitions~(\ref{def:disj}-\ref{def:forall}), except for Property~(\ref{all-arrow2}) in Definition~\ref{def:forall}. This   sublogic cannot be defined under the usual  set-theoretical semantics since  Property~(\ref{all-arrow2})  is \emph{not independent} from Definition~\ref{def:exist} and Property~(\ref{all-arrow1}) in this setting. Indeed, for every interpretation $\mathcal{I}$ it holds that if ${R'}^\mathcal{I}\subseteq R^\mathcal{I}$, $\mathsf{dom}(R')^\mathcal{I}\subseteq \forall R.C^\mathcal{I}$ and $x\in \mathsf{cod}(R')^\mathcal{I}$, then $x\in C^\mathcal{I}$.

First, we need to show the independence of  Property~(\ref{all-arrow2}), i.e. that it is not derived from the other properties introduced in Definitions~(\ref{def:disj}-\ref{def:exist}) and from Property~(\ref{all-arrow1}) in Definition~\ref{def:forall}. For this purpose, we present in Example~\ref{ex:notforall}  a category   which verifies all properties from  Definition~\ref{def:exist}, and Property~(\ref{all-arrow1}), but not Property~(\ref{all-arrow2}).

\begin{example}\label{ex:notforall} Every ontology category $\mathscr{C}_c\tuple{\forall R.D,\emptyset}$     that has an object $\forall R.D$  has to consist at least of  the following arrows : $\forall R.D \longleftrightarrows\neg \exists R.\neg D$, $\exists R.\neg D\longleftrightarrows \mathsf{dom}(R_{(\exists R.\neg D)})$  where $R$ and $R_{(\exists R.\neg D)}$ are objects of $\mathscr{C}_r\tuple{\forall R.D,\emptyset}$. By applying Definitions~\ref{def:syntax-cat}-\ref{def:exist}, we can add to  $\mathscr{C}_c\tuple{\forall R.D,\emptyset}$ other arrows such as $\exists R.\neg D \sqcap \neg \exists R.\neg D\longrightarrow \bot$, $\top \longrightarrow  \exists R.\neg D \sqcup \neg \exists R.\neg D$ and $\mathsf{cod}(R_{(\exists R.\neg D)})\longrightarrow \neg D$. 

Let $R'$ be some object in $\mathscr{C}_c\tuple{\forall R.D,\emptyset}$   such that $R'\longrightarrow R$ and $\mathsf{dom}(R')\longrightarrow \forall R.D$. If we  apply exhaustively Definitions~\ref{def:syntax-cat}-\ref{def:exist} and Property~(\ref{all-arrow1}) in Definition~\ref{def:forall}, they never add  to $\mathscr{C}_c\tuple{\forall R.D,\emptyset}$ an arrow   $\mathsf{cod}(R')\longrightarrow  D$, which should be derived by  Property~(\ref{all-arrow2}) in Definition~\ref{def:forall}. 
%Note that this is possible if $\mathcal{O}$ has an axiom such as $\mathsf{cod}(R)\sqsubseteq D$.
\end{example}

We can follow the same idea used in Definition~\ref{def:cat-satisfiability} to introduce  concept unsatisfiability for   $\mathcal{ALC}_{\overline{\forall}}$ without referring to set interpretation. 
%In Definition~\ref{def:unsat-forall}, it is required that ontology category $\mathscr{C}_c\langle C, \mathcal{O}\rangle$ should be smallest by set inclusion order over $\mathsf{Ob}(\mathscr{C}_c)$ and $\mathsf{Hom}(\mathscr{C}_c)$. This makes completeness of the algorithm possible to be established without harming soundness of the algorithm.

\begin{definition}[Category-theoretical unsatisfiability in $\mathcal{ALC}_{\overline{\forall}}$]\label{def:concept-sat}
Let $C$ be an $\mathcal{ALC}_{\overline{\forall}}$ concept, $\mathcal{O}$ an $\mathcal{ALC}_{\overline{\forall}}$ ontology.  $C$ is category-theoretically unsatifiable with respect to $\mathcal{O}$ if  there is an ontology category $\mathscr{C}_c\langle C, \mathcal{O}\rangle$ which has an arrow $C\longrightarrow \bot$.
\end{definition}

In the sequel, we propose an algorithm for checking satisfiability of an $\mathcal{ALC}_{\overline{\forall}}$ concept $C_0$ with respect to an  ontology $\mathcal{O}$. Similarly to the usual set theoretical-semantics, we can define  Negation Normal Form (NNF) $\mathsf{NNF}(C)$  of a concept object $C$, i.e. negations appear only in front of concept name object. It is possible to convert polynomially a concept object to its NNF by using the following properties resulting from Lemma~\ref{lem:neg-prop} and Definition~\ref{def:forall}.
\begin{enumerate}
    \item[] $\neg(C\sqcup D)\longleftrightarrows \neg C\sqcap \neg D$
    \item[] $\neg(C\sqcap D)\longleftrightarrows \neg C\sqcup \neg D$
    \item[] $\neg \forall R.C \longleftrightarrows \exists R.\neg C$
    \item[] $\neg \exists R.C \longleftrightarrows \forall R.\neg C$
    \item[] $\neg \neg C\longleftrightarrows  C$
    \item[] $\neg \top\longleftrightarrows  \bot$
    \item[] $\neg \bot\longleftrightarrows  \top$
\end{enumerate}

\begin{definition}[subconcepts]\label{def:subconcepts} Let $C_0$ be an $\mathcal{ALC}$ concept and $\mathcal{O}$ an $\mathcal{ALC}_{\overline{\forall}}$ ontology.
 A smallest set $\mathsf{sub}\langle C_0,\mathcal{O}\rangle$ of subconcepts   occurring in $C_0$ and $\mathcal{O}$ is defined as follows:
    \begin{enumerate}
    
    \item $C_0 \in\mathsf{sub}\langle C_0,\mathcal{O}\rangle$, and  $E,  F \in\mathsf{sub}\langle C_0,\mathcal{O}\rangle$ for each axiom $E\sqsubseteq F\in \mathcal{O}$;

    \item  If $E \sqcap F$ or  $E \sqcup F\in\mathsf{sub}\langle C_0,\mathcal{O}\rangle$ then $E, F\in \mathsf{sub}\langle C_0,\mathcal{O}\rangle$; 
    
    \item  If $\exists R.D $ or  $\forall R.D \in\mathsf{sub}\langle C_0,\mathcal{O}\rangle$ then
    $D\in \mathsf{sub}\langle C_0,\mathcal{O}\rangle$;
    
    \item If $C \in\mathsf{sub}\langle C_0,\mathcal{O}\rangle$ then $\mathsf{NNF}(\neg C)\in \mathsf{sub}\langle C_0,\mathcal{O}\rangle$.
\end{enumerate}

\end{definition}

To check satisfiability of an $\mathcal{ALC}_{\overline{\forall}}$ concept $C_0$ with respect to an  ontology $\mathcal{O}$, we initialize an ontology category $\mathscr{C}\langle C_0,\mathcal{O} \rangle$ and saturate it by applying the saturation rules in Table~\ref{tab:satRules} to $\mathscr{C}_c\langle C_0,\mathcal{O} \rangle$ until no rule is applicable. Each saturation rule in Table~\ref{tab:satRules} corresponds to a property given in the definitions of constructor semantics. Moreover, we adopt the following assumptions when designing the algorithm. 

\begin{enumerate}
    \item The collections $\mathsf{Ob}(\mathscr{C})$ and $\mathsf{Hom}(\mathscr{C})$ of  a category $\mathscr{C}$ built by the algorithm are considered as sets, i.e there are no duplicate in them.  
    
    \item When the algorithm adds an object $X$ to a category $\mathscr{C}$, it adds also an identity arrow $X\longrightarrow X$, and $\bot\longrightarrow X$, $X\longrightarrow \top$ where $\bot,\top$ are initial and terminal objects of  $\mathscr{C}$. In particular, if an object $X$ is added  to a concept ontology category, it adds  $\mathsf{NNF}(X)$ as well. For the sake of simplicity, we will not mention explicitly these arrows in the rules described in Table~\ref{tab:satRules}. 
\end{enumerate}

\begin{table}
    \centering
    \begin{tabular}{|p{12cm}|}
    \hline
\textbf{[$\bot$-rule]} 
 If $X\longrightarrow\mathsf{NNF}(\neg C)$, $X\longrightarrow C$ are arrows of $\mathscr{C}_c$,  and $X\longrightarrow \bot$ is not an arrow of $\mathscr{C}_c$, then we add an object $C \sqcap \mathsf{NNF}(\neg C)$ and   an arrow   $X\longrightarrow\bot$ to $\mathscr{C}_c$.\\
 \hline
\textbf{[$\bot_m$-rule]}  If there is an arrow   $E\sqcap F\longrightarrow \bot$ in $\mathscr{C}_c$, and  there is no arrow  $E \longrightarrow  \mathsf{NNF}(\neg F)$ or $F \longrightarrow  \mathsf{NNF}(\neg E)$ in  $\mathscr{C}_c$, then we add  arrows $E \longrightarrow  \mathsf{NNF}(\neg F)$ and $F \longrightarrow  \mathsf{NNF}(\neg E)$    to $\mathscr{C}_c$.
\\
 \hline
\textbf{[$\top$-rule]}
 If $\mathsf{NNF}(\neg C)\longrightarrow X$, $C\longrightarrow X$ are arrows of $\mathscr{C}_c$,  and $\top \longrightarrow X$ is not an arrow of $\mathscr{C}_c$, then we add an object $C \sqcup \mathsf{NNF}(\neg C)$ and an arrow  $\top \longrightarrow X$ to $\mathscr{C}_c$.
\\
 \hline
\textbf{[$\top_m$-rule]} If there is an arrow   $\top\longrightarrow E\sqcup F$ in $\mathscr{C}_c$, and  there is no arrow  $\mathsf{NNF}(\neg F) \longrightarrow  E$ or $\mathsf{NNF}(\neg E) \longrightarrow  F$ in  $\mathscr{C}_c$, then we add  arrows $ \mathsf{NNF}(\neg F) \longrightarrow E$ and $ \mathsf{NNF}(\neg E) \longrightarrow F$    to $\mathscr{C}_c$.
\\
\hline
\textbf{[$\neg$-rule]} If there is an arrow  $C \longrightarrow D$ in $\mathscr{C}_c$ and there is no arrow $\mathsf{NNF}(\neg D)  \longrightarrow \mathsf{NNF}(\neg C)$   in $\mathscr{C}_c$, then we add objects $\mathsf{NNF}(\neg C), \mathsf{NNF}(\neg D)$ and an arrow  $\mathsf{NNF}(\neg D)  \longrightarrow \mathsf{NNF}(\neg C)$  to  $\mathscr{C}_c$.\\
\hline
%\textbf{[$\neg_\top$-rule]} If there is an arrow  $\mathsf{NNF}(\neg C) \longrightarrow D$ in $\mathscr{C}_c$ and there is no arrow $\mathsf{NNF}(\neg D) \longrightarrow C$   in $\mathscr{C}_c$, then we add objects $C,\mathsf{NNF}(\neg D)$ and an arrow  $\mathsf{NNF}(\neg D) \longrightarrow C$  to  $\mathscr{C}_c$.\\
%\hline
\textbf{[$\sqcap$-rule]} If there is an object $E \sqcap F$ of $\mathscr{C}_c$ and there are not arrows $E \sqcap F\longrightarrow E$ and $E \sqcap F\longrightarrow F$, then we add objects $E,F$ and arrows $E \sqcap F\longrightarrow E$, $E \sqcap F\longrightarrow F$ to  $\mathscr{C}_c$.\\
\hline
\textbf{[$\sqcap_m$-rule]} If there is an object $E\sqcap F$ and arrows  $X\longrightarrow E$,  $X\longrightarrow F$ in $\mathscr{C}_c$, and there is no arrow $X \longrightarrow E \sqcap F$ in $\mathscr{C}_c$, then we add  an arrow $X \longrightarrow E \sqcap F$ to $\mathscr{C}_c$.\\
\hline
\textbf{[$\sqcup$-rule]} If there is an object $E \sqcup F$ in $\mathscr{C}_c$ and there are not two arrows $E\longrightarrow E \sqcup F$ and $F\longrightarrow E \sqcup F$, then we add objects $E,F$ and  arrows $E\longrightarrow E \sqcup F$,  $F\longrightarrow E \sqcup F$ to  $\mathscr{C}_c$.\\
\hline
\textbf{[$\sqcup_m$-rule]} If there is an object $E\sqcup F$ and  arrows  $E\longrightarrow X$, $F\longrightarrow X$ in $\mathscr{C}_c$, and there is no arrow  $E \sqcup F \longrightarrow X$  in $\mathscr{C}_c$, then we add   an arrow $E \sqcup F \longrightarrow X$ to $\mathscr{C}_c$.\\
\hline
%\textbf{[$\sqcup_\mathsf{ind}$-rule]} If there is an object $C \sqcap   (D\sqcup E)$ in  $\mathscr{C}_c$, and both $C \sqcap D$ and $C \sqcap E$ are not objects of $\mathscr{C}_c$, then we choose $Y$ from $\{D, E\}$, and add an object $C\sqcap Y$ to $\mathscr{C}_c$.\\ %and an  arrow $C \sqcap   (D\sqcup E) \longrightarrow C\sqcap Y$ to $\mathscr{C}_c$.\\
%\hline
%\textbf{[$\sqcup_\mathsf{dis}$-rule]}  If there are    objects $C \sqcap   (D\sqcup E)$, $C\sqcap D$,  $C\sqcap E$ and arrows $C\sqcap D\longrightarrow W$, $C\sqcap E\longrightarrow W$ in $\mathscr{C}_c$, and there is no arrow  $C \sqcap   (D\sqcup E) \longrightarrow W$  in $\mathscr{C}_c$, then we  add  an arrow $C \sqcap   (D\sqcup E) \longrightarrow W$ to $\mathscr{C}_c$.\\
\textbf{[$\sqcup_\mathsf{dis}$-rule]}  If there are two objects $C \sqcap   (D\sqcup E)$ and $W$ in $\mathscr{C}_c$, and there is no arrow  $C \sqcap   (D\sqcup E)\longrightarrow W$ in $\mathscr{C}_c$, and  $\mathsf{check}(\mathscr{C}_c,C \sqcap   (D\sqcup E),W)$ returns ``$\mathsf{true}$", then we add an arrow  $C \sqcap   (D\sqcup E)\longrightarrow W$ to $\mathscr{C}_c$.\\
\hline
\textbf{[$\exists$-rule]} If there is an object $\exists R.D$ of $\mathscr{C}_c$, there is no object $R_{(\exists R.D)}$ of $\mathscr{C}_r$, then we add objects $R_{(\exists R.D)}, R$,  an arrow $R_{(\exists R.D)}\longrightarrow R$ to $\mathscr{C}_r$,  and add to  $\mathscr{C}_c$ objects $D$,  $\mathsf{cod}(R_{(\exists R.D)})$, $\mathsf{dom}(R_{(\exists R.D)})$, arrows $\mathsf{cod}(R_{(\exists R.D)})\longrightarrow D$,
$\mathsf{dom}(R_{(\exists R.D)})\longleftrightarrows \exists R.D$.\\
\hline
\textbf{[$\exists_m$-rule]} If there is an object $\exists R.D$ of $\mathscr{C}_c$ and  an object $R'$ of $\mathscr{C}_r$  such that $R' \longrightarrow R$ and $\mathsf{cod}(R') \longrightarrow D$, and there is no arrow $\mathsf{dom}(R')\longrightarrow\mathsf{dom}(R_{(\exists R.D)})$,  then we add to $\mathscr{C}_c$  an arrow   $\mathsf{dom}(R') \longrightarrow \mathsf{dom}(R_{(\exists R.D)})$.\\
\hline
\textbf{[$\forall$-rule]} If  there is an object  $\forall R.C$  in $\mathscr{C}_c$ and there are no arrows $\forall R.C \longleftrightarrows \mathsf{NNF}(\neg \exists R.\neg C)$ in $\mathscr{C}_c$, then we add objects $\exists R.\mathsf{NNF}(\neg C)$, $\mathsf{NNF}(\neg \exists R.\neg C)$ and  arrows   $\forall R.C \longleftrightarrows \mathsf{NNF}(\neg \exists R.\neg C)$ to $\mathscr{C}_c$. \\
\hline
%\textbf{[$\forall_m$-rule]} If  there are arrows $X\longrightarrow \exists R.C$, $X\longrightarrow \forall R.D$, $C\longrightarrow C'$, $D\longrightarrow \mathsf{NNF}(\neg C') $ in $\mathscr{C}_c$ , and there is no arrow $X\longrightarrow \bot$ in $\mathscr{C}_c$, then   we add to $\mathscr{C}_c$  an arrow   $X\longrightarrow \bot$. \\
%\hline
\textbf{[$\mathsf{trans}$-rule]} If there are objects $X,Y,Z$ and arrows $X\longrightarrow Y$,  $Y\longrightarrow Z$ of  $\mathscr{C}_c$, there is no arrow  $X\longrightarrow Z$ of  $\mathscr{C}_c$, then we add an arrow  $X\longrightarrow Z$ to $\mathscr{C}_c$.\\
\hline
\textbf{[$\mathsf{dc}$-rule]} 
\noindent \textbf{(a)} 
If $R$ is an object of $\mathscr{C}_r$,
and  $X\longrightarrow\bot$  is  an arrow of $\mathscr{C}_c$ for some $X\in \{\mathsf{dom}(R), \mathsf{cod}(R)\}$, and $Y\longrightarrow\bot$  is not  an arrow of $\mathscr{C}_c$ for some $Y\in \{\mathsf{dom}(R), \mathsf{cod}(R)\}$,  then we add 
$Y\longrightarrow\bot$ to $\mathscr{C}_c$.\\
\iffalse
$\textbf{(b)}$ If there is an arrow $R\longrightarrow R'$ in $\mathscr{C}_r$ and no arrows
$\mathsf{dom}(R)\longrightarrow\mathsf{dom}(R')$, $\mathsf{cod}(R)\longrightarrow\mathsf{cod}(R')$ in $\mathscr{C}_c$, then add arrows $\mathsf{dom}(R)\longrightarrow\mathsf{dom}(R')$, $\mathsf{cod}(R)\longrightarrow\mathsf{cod}(R')$ to $\mathscr{C}_c$.  
\fi
\noindent\textbf{(b)} If $R\longrightarrow R'$ is an arrow  of $\mathscr{C}_r$ and ${X\longrightarrow Y}$ is not an arrow of $\mathscr{C}_c$, for some ${(X, Y)\in \left\{\left(\mathsf{dom}(R),
\mathsf{dom}(R')\right), \left(\mathsf{cod}(R), \mathsf{cod}(R')\right)\right\}}$,  then we add to  $\mathscr{C}_c$ an arrow $X\longrightarrow Y$ to $\mathscr{C}_c$.\\
\hline
 \end{tabular}
    \caption{Saturation rules}
    \label{tab:satRules}
\end{table}

\begin{algorithm}[!h]
\SetKwInOut{Input}{Input}\SetKwInOut{Output}{Output}
\Input{$\langle C_0,  \mathcal{O}\rangle$ where $C_0$ is an $\mathcal{ALC}_{\overline{\forall}}$ concept and $\mathcal{O}$  an ontology} 
 
\Output{$\mathsf{true}$ or $\mathsf{false}$}
\BlankLine
Initialize two categories $\mathscr{C}_c$ and $\mathscr{C}_r$\;

Add  an object $C_0$  to $\mathscr{C}_c$\label{algo:addC}\;
\ForEach{$E\sqsubseteq F$ of $\mathcal{O}$\label{algo:l0}}{
 Add  objects $E, F$ and an arrow  $E \longrightarrow  F$ to   $\mathscr{C}_c$\;
}\label{algo:end-l0}

\While{there is a saturation rule $\mathbf{r}$ that is applicable to $\mathscr{C}_c$ and $\mathscr{C}_r$\label{algo:l1} }{
   Apply $\mathbf{r}$ to $\mathscr{C}_c$ and $\mathscr{C}_r$\;
   
}
\If{there is an arrow $C_0\longrightarrow \bot$ in $\mathscr{C}_c$\label{algo:c1}}{
       \textbf{return}$~\mathsf{false}$\;
   }\Else{
     \textbf{return}$~\mathsf{true}$\; 
   }
\caption{$\mathsf{isSatisfiable}\langle C_0,\mathcal{O} \rangle$ }\label{algo:consistency}
\end{algorithm}

\begin{algorithm}[!h]
\SetKwInOut{Input}{Input}\SetKwInOut{Output}{Output}
\Input{$\langle \mathscr{C}_c, (C_1\sqcup D_1)\sqcap \cdots \sqcap (C_n\sqcup D_n),W\rangle$   where $\mathscr{C}_c$  is an category and $(C_1\sqcup D_1)\sqcap \cdots \sqcap (C_n\sqcup D_n)$ and $W$ are  objects  of  $\mathscr{C}_c$} 
 
\Output{$\mathsf{true}$ or $\mathsf{false}$}
\BlankLine

%$\Gamma\gets \emptyset$\; 
\ForEach{$(X_1,\cdots,X_n)$ with $X_i\in \{C_i,D_i\}$\label{algo:check:l1}}{
 %\ForEach{object $W$ of $\mathscr{C}_c$}{
 \If{there is no object $(X_{i_1}\sqcap \cdots \sqcap X_{i_k})$ in  $\mathscr{C}_c$ with $1\leq i_j\leq n$ such that there is an arrow $(X_{i_1}\sqcap \cdots \sqcap X_{i_k})\longrightarrow W$ in $\mathscr{C}_c$}{
    \textbf{return~false} \;
 }
 %}
} 
\textbf{return~true} \;
\caption{$\mathsf{check}\langle \mathscr{C}_c,(C_1\sqcup D_1)\sqcap \cdots \sqcap (C_n\sqcup D_n), W\rangle$ }\label{algo:check}
\end{algorithm}

\begin{lemma}[Complexity]\label{lem:complexity} Let $C$ be an $\mathcal{ALC}_{\overline{\forall}}$ concept and   $\mathcal{O}$ an ontology. Algorithm~\ref{algo:consistency}   runs in   polynomial space for the input $\tuple{C,\mathcal{O}}$.
\end{lemma}

\begin{proof}  Since  there are at most two arrows between two objects, it 
suffices to determine the number of different objects added to $\mathscr{C}_c$
and $\mathscr{C}_r$. For this purpose, we analyse the behavior of each rule in
Table~\ref{tab:satRules} and determine the number of objects added by each 
such a rule. 

Let $n=\norm{\tuple{C_0,\mathcal{O}}}$ be the size of input,
\emph{i.e.} the number of bytes. By Definition~\ref{def:subconcepts}, we write
$\mathsf{sub}\tuple{C_0, \mathcal{O}}$, the set of sub-concepts. For 
Definition~\ref{def:subconcepts}.1, 3 and 4 they are sub-strings of 
$\tuple{C, \mathcal{O}}$, not all of sub-strings belongs to it, so we have 
at most $\ell\leq n(n+1)/2$ objects created from those three 
properties. From Definition~\ref{def:subconcepts}.2, we simply double
the number of object obtained from the other three point - note that 
$\mathsf{NNF}(\neg\neg C)$ \emph{is} $C$ - so in total $\norm{\mathsf{sub}
\tuple{C_0,\mathcal{O}}} = 2\ell = k \leq O(n^2)$. It should be noted that
every time we add a object to $\Catc$, we also add its negation, we won't 
mention it every time an object is added but it is done. Hence why we 
reason on $k$ instead of $\ell$ even when we index added objects on 
a object $C$ of $\mathsf{sub}\tuple{C_0,\mathcal{O}}$.

\begin{enumerate}
    \item \textbf{[Initialisation]}(line 3-5). It is the first loop of the 
    algorithm and corresponds to Definition~\ref{def:subconcepts}.1, 
    in the sense the only added objects by the algorithm at this stage are
    already part of $\mathsf{sub}\tuple{C_0,\mathcal{O}}$, \emph{i.e.} it 
    cannot add more object 
    than $\norm{\mathsf{sub}\tuple{C_0, \mathcal{O}}} = k\leq O(n^2)$ to
    $\Catc$.
    
    \item \textbf{[$\sqcap$-rule]}, \textbf{[$\sqcup$-rule]}. These two rules
    correspond to Definition~\ref{def:subconcepts}.3, objects added 
    to $\Catc$ are already part of $\mathsf{sub}\tuple{C_0,\mathcal{O}}$, 
    meaning there are $O(n^2)$ of them.
    
    \item \textbf{[$\exists$-rule]}. Objects $D$ added to $\Catc$ by this rule
    are already part of $\mathsf{sub}\tuple{C_0, \mathcal{O}}$ according to 
    Definition~\ref{def:subconcepts}.4, following the same reasoning as for the 
    previous rules, we have $O(n^2)$ of them. Remaining objects to be
    added to $\Catc$, namely $\mathsf{cod}(R_{(\exists R.D)})$ and 
    $\mathsf{dom}(R_{(\exists R.D)})$, are \emph{not} part of 
    $\mathsf{sub}\tuple{C_0, \mathcal{O}}$. However, they are only added once 
    per each occurrence of object of the form $\exists R.D$ in $\mathsf{sub}
    \tuple{C_0, \mathcal{O}}$, \emph{i.e.} the number of  object added is bounded by 
    $2k$. In total, this rule add a number of to $\Catc$ object smaller 
    than $3k\leq O(n^2)$. We leave the number of objects added to 
    $\mathscr{C}_R$ for later.

    \item \textbf{[$\forall$-rule]}. From 
    Definition~\ref{def:subconcepts}.4, $\mathsf{NNF}(\neg\forall R.D)
    \longleftrightarrows\exists R.\mathsf{NNF}(\neg D)$ is already part of 
    $\mathsf{sub}\tuple{C_0, \mathcal{O}}$. Only $\exists R.\neg D$ isn't 
    part of Definition~\ref{def:subconcepts}, however, it is linearly 
    bounded by the number of object of the form $\forall R.D$ which is
    smaller than $k$, hence this rule add at most $2k\leq O(n^2)$ objects 
    to $\Catc$ in total.

    \item \textbf{[$\bot$-rule]}, \textbf{[$\top$-rule]}. These rules
    only add one object of the type $C\sqcap\mathsf{NNF}(\neg C)$ or 
    $C\sqcup\mathsf{NNF}(\neg C)$ to $\Catc$ per object $C$, respectively. 
    Assume that we have applied the rule once on a single $X$ and added 
    $C\sqcap\mathsf{NNF}(\neg C)$ or $C\sqcup\mathsf{NNF}(\neg C)$ to $\Catc$. 
    If there is an other $Y$ such that $Y\longrightarrow C$ and 
    $Y\longrightarrow\mathsf{NNF}(\neg C)$ or $C\longrightarrow Y$ and 
    $\mathsf{NNF}(\neg C)\longrightarrow Y$, we can apply the rules nonetheless, 
    however, it won't add another occurrence of $C\sqcap\mathsf{NNF}(\neg C)$ or 
    $C\sqcup\mathsf{NNF}(\neg C)$ since we are dealing with sets so an 
    object can appear one and only one time - and there's already an 
    occurrence of it since we have applied the rule before. So the number of 
    object added by these rules is bounded by $k\leq O(n^2)$.
    
    \item \textbf{[$\neg$-rule]}. If we have an arrow $C\longrightarrow D$
    in $\Catc$, $C$ and $D$ are already in $\Catc$, consquently $\mathsf{NNF}
    (\neg C)$ and $\mathsf{NNF}(\neg D)$ as well. Since we are dealing with 
    sets, we are not adding any new objects.
    
    \item \textbf{[$\mathsf{trans}$-rule]}, \textbf{[$\sqcup_m$-rule]}, 
    \textbf{[$\sqcap_m$-rule]}, \textbf{[$\sqcup_{dis}$-rule]},
    \textbf{[$\bot_m$-rule]}, \textbf{[$\top_m$-rule]},\\
    \textbf{[$\exists_m$-rule]}, \textbf{[$\mathsf{dc}$-rule]}. These rules only 
    deal with arrows and no objects are added - we can ignore them.
    \item Algorithm~\ref{algo:check} ($\mathsf{check}$). This algorithm can be  implemented such that it runs in  polynomial space. For example, we can consider each vector $(X_1,\cdots,X_n)$  as a binary number $(b(X_1),\cdots,b(X_n))$ where $b(X_i)=1$ iff $X_i=C_i$ ($X_i\in \{C_i,D_i\}$). In this case, the loop from Line~\ref{algo:check:l1} can take binary numbers from $0$ to $2^{n}-1$, from which it can determine the vector $(X_1,\cdots,X_n)$.  
\end{enumerate}

For \textbf{[$\exists$-rule]} the number of objects $R_{(\exists R.D)}$
added to $\mathscr{C}_r$ is bounded by $\ell$. Objects $R$ are only added 
to $\mathscr{C}_r$ once per role name appearing in $\mathcal{T}$, it is 
also bounded by $\ell$. In summary, each rule adds to $\Catc$ a number of
objects bounded by $O(n^2)$, the number of objects of $\mathscr{C}_r$ is 
also bounded by 
$O(n^2)$, thus Algorithm~\ref{algo:consistency} runs in 
polynomial space.$\hfill\square$
\end{proof}

\begin{lemma}[soundness and completeness]\label{lem:sndn-comp} Let $C_0$ be an $\mathcal{ALC}_{\overline{\forall}}$ concept and   $\mathcal{O}$   an  ontology. Algorithm~\ref{algo:consistency}   returns ``$\mathsf{false}$" with the input $\tuple{C_0,\mathcal{O}}$ iff $C_0$ is unsatisfiable with respect to $\mathcal{O}$.
\end{lemma}

%When Algorithm~\ref{algo:consistency}  uses the non-deterministic rules, there may be different runs which yield  different outputs. Each such a run corresponds to a choice performed by  the non-deterministic rule \textbf{[$\sqcup_\mathsf{ind}$-rule]}. We show that soundness and completeness of Algorithm~\ref{algo:consistency} using the non-deterministic rules can be obtained from Lemma~\ref{lem:sndn-comp}. Indeed, if  Algorithm~\ref{algo:consistency} using the non-deterministic rules returns ``$\mathsf{true}$" for a run, then by Lemma~\ref{lem:sndn-comp} Algorithm~\ref{algo:consistency} using the deterministic rules returns ``$\mathsf{true}$". By Lemma~\ref{lem:equiv}, $C_0$ is satisfiable with respect to   $\mathcal{O}$. Conversely, assume that $C_0$ is satisfiable with respect to   $\mathcal{O}$. By Lemma~\ref{lem:equiv}, Algorithm~\ref{algo:consistency} using the deterministic rules returns ``$\mathsf{true}$". Due to Lemma~\ref{lem:sndn-comp}, there is a run of Algorithm~\ref{algo:consistency} using the non-deterministic rules which returns ``$\mathsf{true}$".
\begin{proof} ``$\Longleftarrow$". Let $\mathscr{C}_c$ be a category built by Algorithm~\ref{algo:consistency}  that returns ``$\mathsf{false}$". This implies that $\mathscr{C}_c$ has an arrow $C_0\longrightarrow \bot$. First, we  show that $\mathscr{C}_c$  satisfies all properties in Definitions~\ref{def:syntax-cat}-\ref{def:exist}  and Property~(\ref{all-arrow1}) in Definition~\ref{def:forall}.

\begin{enumerate}
    \item  Definition~\ref{def:syntax-cat}. When a saturating rule in Table~\ref{tab:satRules} adds a new  concept (resp. role) object, it adds also an identity arrow and arrows between the object and  $\bot$ (resp. $R_\bot$) and $\top$ (resp. $R_\top$). Thus, Properties~\ref{def:syntax-cat:1} and \ref{def:syntax-cat:2} in Definition~\ref{def:syntax-cat} hold in $\mathscr{C}_c$. Moreover, when Algorithm~\ref{algo:consistency} returns ``$\mathsf{true}$",  the loop from Line~\ref{algo:l1} terminates, and thus no rule is applicable. That means \textbf{[$\mathsf{trans}$-rule]} is not applicable. This implies that Property~\ref{def:syntax-cat:3}   in Definition~\ref{def:syntax-cat} holds in $\mathscr{C}_c$. Analogously,  \textbf{[$\mathsf{dc}$-rule]} is not applicable when Algorithm~\ref{algo:consistency} returns ``$\mathsf{true}$". Hence, Property~\ref{def:syntax-cat:4}   in Definition~\ref{def:syntax-cat} is ensured in $\mathscr{C}_c$,  and thus   in $\mathscr{C}'_c$. 
    
    \item  Definition~\ref{def:onto-cat}. By the assumption, for each object $X$ added to  $\mathscr{C}_c$, there is also $\mathsf{NNF}(\neg X)$ in $\mathscr{C}_c$. In this case,  non applicability of \textbf{[$\bot$-rule]} and \textbf{[$\top$-rule]}  guarantees that there are arrows $\mathsf{NNF}(\neg X)\sqcap X\longrightarrow \bot$, $\top \longrightarrow \mathsf{NNF}(\neg X)\sqcup X$ in   $\mathscr{C}_c$. Let $C\sqcap X\longrightarrow \bot$ (resp. $\top \longrightarrow C\sqcup X$) be an arrow in $\mathscr{C}_c$. Non applicability of \textbf{[$\bot_m$-rule]} and \textbf{[$\top_m$-rule]} implies that there are arrows $C\longrightarrow \mathsf{NNF}(\neg X)$, $X\longrightarrow \mathsf{NNF}(\neg C)$ (resp. $\mathsf{NNF}(\neg X)\longrightarrow C $, $\mathsf{NNF}(\neg C)\longrightarrow X$) in $\mathscr{C}_c$.
    %and thus   in $\mathscr{C}'_c$. 
    
    \item  Definition~\ref{def:disj}.  Let $E, F$ and $E\sqcup F$ be objects of 
    $\mathscr{C}_c$. Thanks to \textbf{[$\sqcup$-rule]},  $\mathscr{C}_c$ has
    arrows $E\longrightarrow E\sqcup F$ and $F\longrightarrow E\sqcup F$ which  
    ensure Property~(\ref{disj01}) in  Definition~\ref{def:disj}.  Let  $E\longrightarrow X$ and $F\longrightarrow X$ be  arrows in $\mathscr{C}_c$, and $E\sqcup F$ be an object of $\mathscr{C}_c$.   Non-applicability of  \textbf{[$\sqcup_m$-rule]} ensures that there is  an arrow $E\sqcup F\longrightarrow X$ in $\mathscr{C}_c$. Hence,  Property~(\ref{disj02}) in  Definition~\ref{def:disj} holds. %Hence, it holds also  in $\mathscr{C}'_c$. 
     
    \item     Definition~\ref{def:conj}.  Let $E, F$ and $E\sqcap F$ be objects of 
    $\mathscr{C}_c$. Thanks to \textbf{[$\sqcap$-rule]}, we  have 
    arrows $E\sqcap F\longrightarrow E$ and $E\sqcap F\longrightarrow F$ which      ensures Property~(\ref{conj01}) in  Definition~\ref{def:conj}. Assume that  $X\longrightarrow E$ and $X\longrightarrow F$ be  arrows in $\mathscr{C}_c$, and $E\sqcap F$ be an object of $\mathscr{C}_c$.   Non-applicability of  \textbf{[$\sqcap_m$-rule]} ensures that there is  an arrow $X\longrightarrow E\sqcap F$ in $\mathscr{C}_c$. Hence,  Property~(\ref{conj1}) in  Definition~\ref{def:conj} holds. Assume that  $C\sqcap (D\sqcup E)$, and $(C\sqcap D) \sqcup (C\sqcap E)$ be  objects in $\mathscr{C}_c$.    Due to  non-applicability of  \textbf{[$\sqcup$-rule]},  $\mathscr{C}_c$ has objects $C\sqcap D$, $C\sqcap E$  and  arrows $C\sqcap D \longrightarrow (C\sqcap D) \sqcup (C\sqcap E)$ and  $C\sqcap E\longrightarrow (C\sqcap D) \sqcup (C\sqcap E)$, i.e $\mathsf{check}(\mathscr{C}_c, C\sqcap (D\sqcup E), (C\sqcap D) \sqcup (C\sqcap E))$ returns ``$\mathsf{true}$". Due to non-applicability of  \textbf{[$\sqcup_{\mathsf{dis}}$-rule]},  $\mathscr{C}_c$ has an arrow $C\sqcap (D\sqcup E) \longrightarrow (C\sqcap D) \sqcup (C\sqcap E)$. Hence,    Property~(\ref{conj2}) in  Definition~\ref{def:conj} holds.

    \item   Definition~\ref{def:neg}.  Let $C,\mathsf{NNF}(\neg C), C\sqcap \mathsf{NNF}(\neg C), C\sqcup \mathsf{NNF}(\neg C), C\sqcap X, C\sqcup X$ be objects of     $\mathscr{C}_c$. Due to non-applicability of    \textbf{[$\bot$-rule]} and \textbf{[$\top$-rule]}, $\mathscr{C}_c$ has arrows $C\sqcap \mathsf{NNF}(\neg C)\longrightarrow \bot$, $\top \longrightarrow  C\sqcup \mathsf{NNF}(\neg C)$. Assume that   $\mathscr{C}_c$ has an arrow $C\sqcap X\longrightarrow \bot$ (resp. $\top \longrightarrow  C\sqcup X$). Due to non-applicability of    \textbf{[$\bot_m$-rule]} and \textbf{[$\top_m$-rule]}, $\mathscr{C}_c$ has arrows $C \longrightarrow \mathsf{NNF}(\neg X)$ (resp. $\mathsf{NNF}(\neg X) \longrightarrow  C$).  Hence,    Properties~(\ref{neg-bot})-(\ref{neg-min}) in  Definition~\ref{def:neg} hold.

    \item   Definition~\ref{def:exist}. Let $\exists R.C$ be an object of $\mathscr{C}_c$. Due to non-applicability of \textbf{[$\exists$-rule]}, it follows that $R_{(\exists R.C)},R$ are  objects, and  $R_{(\exists R.C)}\longrightarrow R$ is an arrow of $\mathscr{C}_r$; and $\mathsf{dom}(R_{(\exists R.C)}),\mathsf{cod}(R_{(\exists R.C)}),C$ are objects, and $\mathsf{dom}(R_{(\exists R.C)})\longleftrightarrows \exists R.C$, $\mathsf{cod}R_{(\exists R.C)})\longrightarrow C$  are arrows of $\mathscr{C}_c$. Let $R'$ be an object of $\mathscr{C}_r$ such that $R'\longrightarrow R$ is an  arrow of $\mathscr{C}_r$, and  $\mathsf{cod}(R')\longrightarrow C$  is an  arrow of $\mathscr{C}_c$. Due to non-applicability of \textbf{[$\exists_m$-rule]}, it follows that $\mathsf{dom}(R')\longrightarrow \mathsf{dom}(R_{(\exists R.C)})$   is an  arrow of $\mathscr{C}_c$. Hence, all properties in  Definition~\ref{def:exist} hold.
    
    \item Property~(\ref{all-arrow1}) in  Definition~\ref{def:forall}. If there is  an object $\forall R.C$ occurring in $\mathscr{C}_c$, it is converted in  $\neg \exists R.\neg C$ due to the assumption that all objects of $\mathscr{C}_c$ are in $\mathsf{NNF}$. Hence, Property~(\ref{all-arrow1}) holds.
\end{enumerate}

We have showed that $\mathscr{C}_c$ satisfy all properties in Definitions~\ref{def:syntax-cat}-\ref{def:exist}  and Property~(\ref{all-arrow1}) in Definition~\ref{def:forall} except that  Property~\ref{conj2} in Definition~\ref{def:conj}. Indeed, if $\mathscr{C}_c$ has an object $C\sqcap (D\sqcup E)$, it may not have an arrow $C\sqcap (D\sqcup E)\ \longrightarrow (C\sqcap D) \sqcup  (D\sqcap E)$ with object $(C\sqcap D) \sqcup  (D\sqcap E)$. 

To complete it, we define a category $\mathscr{C}'_c$ that is   an extension of $\mathscr{C}_c$ by adding objects and arrows  satisfying Property~(\ref{conj2}) in Definition~\ref{def:conj}, i.e if $\mathscr{C}'_c$ has an object $C\sqcap (D\sqcup E)$ then is has also an object $(C\sqcap D) \sqcup  (D\sqcap E)$.  Then we apply saturation rules in Table~\ref{tab:satRules} to $\mathscr{C}'_c$ until no rule is applicable. We can use the same argument above to show that $\mathscr{C}'_c$ satisfy all properties in Definitions~\ref{def:syntax-cat}-\ref{def:exist}  and Property~(\ref{all-arrow1}) in Definition~\ref{def:forall}. Since  $\mathscr{C}_c$ has  an arrow $C_0\longrightarrow \bot$,  $\mathscr{C}'_c$ does.  According to Definition~\ref{def:concept-sat}, $C_0$ is unsatisfiable in $\mathcal{ALC}_{\overline{\forall}}$ with respect to $\mathcal{O}$.

\noindent``$\Longrightarrow$". Assume that $\mathscr{C}_c\tuple{C_0,\mathcal{O}}$, $\mathscr{C}_r\tuple{C_0,\mathcal{O}}$ are \emph{smallest} concept and role ontology categories according to the definitions, i.e they have only necessary objects and arrows which satisfy the dfinition.  By hypothesis, $\mathscr{C}_c\tuple{C_0,\mathcal{O}}$ has an arrow $C_0\longrightarrow \bot$.
We use $\tuple{\mathscr{C}_c,\mathscr{C}_r}$   to denote the concept and role ontology categories built by Algorithm~\ref{algo:consistency}.We define a function $\pi$ from $\mathsf{Ob}(\mathscr{C}_c)\cup \mathsf{Ob}(\mathscr{C}_r)$ to $\mathsf{Ob}(\mathscr{C}_c\tuple{C_0,\mathcal{O}})\cup \mathsf{Ob}(\mathscr{C}_r\tuple{C_0,\mathcal{O}})$. 
Since  $\mathscr{C}_c\tuple{C_0,\mathcal{O}}$  and each $\mathscr{C}_c$ have objects $C_0$, and $E,F$ for all axiom $E\sqsubseteq F$ (Line~\ref{algo:addC} and the loop from Line~\ref{algo:l0}), $\pi$ can be initialized with $\pi(C_0)=C_0$ and  $\pi(E)=E$, $\pi(F)=F$. Thus, the following properties hold for all current $\pi$.
\begin{align}
    &\pi(X)=X \text{ for each object } X  \in \mathsf{Ob}(\mathscr{C}_c) \cup \mathsf{Ob}(\mathscr{C}_r)  \label{claim:1}\\
    &X\longrightarrow Y  \Longleftrightarrow \pi(X)\longrightarrow \pi(Y)\label{claim:2} %\\ 
    %&\pi_{\mathbf{I}_i}(Y_i) \longrightarrow \pi_{\mathbf{I}_i}(W) \text{ for  } 1\leq i\leq 2  \Longrightarrow X\longrightarrow W \in \mathsf{Hom}(\mathcal{C}_{\mathbf{I}_i})    \text{ where } \label{claim:3}\\
    %&\hspace{0.2cm} Y_i \text{ is the non-deterministic of } \mathcal{C}_{\mathbf{I}_i}, X \text{ is the trigger non-deterministic object} \nonumber
\end{align}
%Note that $\pi(X)\longrightarrow \pi(Y)$ means the existence of  objects $X,Y\in \mathsf{Ob}(\mathscr{C}_c)$.  
%First, we observe that if  \textbf{[$\sqcup_\mathsf{det}$-rule]} adds to $\mathscr{C}_c\tuple{C_0,\mathcal{O}}$ an arrow  $C\sqcap (D\sqcup E) \longrightarrow (C\sqcap D)\sqcup (C\sqcap E)$, then  there is also the inverse arrow $C\sqcap (D\sqcup E) \longleftarrow(C\sqcap D)\sqcup (C\sqcap E)$ in $\mathscr{C}_c\tuple{C_0,\mathcal{O}}$.  

We will extend $\pi$ such that  Properties~(\ref{claim:1}) and  (\ref{claim:2})   are preserved for each application of rules. Assume that these properties hold for the current $\pi$.

\begin{enumerate}
    \item \textbf{[$\bot$-rule]}. Assume that $\mathscr{C}_c$ has arrows $X\longrightarrow\mathsf{NNF}(\neg C)$, $X\longrightarrow C$, but it does not have an arrow $X\longrightarrow \bot$. By the hypothesis on   Properties~(\ref{claim:1}) and (\ref{claim:2}), we have  $\pi(X)=X$, $\pi(C)=C$,  $\pi(\mathsf{NNF}(\neg C))=\mathsf{NNF}(\neg C)$,    $\pi(X)\longrightarrow \pi(\mathsf{NNF}(\neg C))$, $\pi(X)\longrightarrow \pi(C)$. In this case,  \textbf{[$\bot$-rule]} adds to
  $\mathscr{C}_c$    objects $C \sqcap \mathsf{NNF}(\neg C)$,  and  arrows  $X\longrightarrow\bot$ and . Since  $\mathscr{C}_c\tuple{C_0,\mathcal{O}}$ is an ontology category, it has $\pi(C \sqcap \mathsf{NNF}(\neg C))$ and $\pi(X)\longrightarrow \pi(\bot)$.
  This implies that Property~(\ref{claim:1}) and the direction ``$\Longrightarrow$" of Property (\ref{claim:2}) are preserved. 
  
  \item \textbf{[$\sqcap$-rule]} and \textbf{[$\sqcup$-rule]}. These rules are applied to objects $C\sqcap D$ and $C\sqcup D$ for adding to  $\mathscr{C}_c$    the following arrows $C \sqcap D\longrightarrow C$, $C \sqcap D\longrightarrow D$, $C\longrightarrow C\sqcup D$, $D\longrightarrow C\sqcup D$.  Since   $\mathscr{C}_c\tuple{C_0,\mathcal{O}}$ is an ontology category, it has  $\pi(C \sqcap D)\longrightarrow \pi(C)$, $\pi(C \sqcap D)\longrightarrow \pi(D)$, $\pi(C)\longrightarrow \pi(C\sqcup D)$, $\pi(D)\longrightarrow \pi(C\sqcup D)$. Thus, Properties~(\ref{claim:1})-(\ref{claim:2}) are preserved.

     \item \textbf{[$\bot_m$-rule]}. Assume that $\mathscr{C}_c$ has an arrow $E\sqcap F\longrightarrow \bot$ and  it does not have an arrow $E \longrightarrow \mathsf{NNF}(\neg F)$ (or $F \longrightarrow \mathsf{NNF}(\neg E)$).  By the hypothesis on   Property~(\ref{claim:1}) and (\ref{claim:2}), $\mathscr{C}_c\tuple{C_0,\mathcal{O}}$ has  an arrow   $\pi(E\sqcap F)\longrightarrow \pi(\bot)$. In this case,  \textbf{[$\bot_m$-rule]} adds to   $\mathscr{C}_c$ an arrow  $E \longrightarrow \mathsf{NNF}(\neg F)$ (or $F \longrightarrow \mathsf{NNF}(\neg E)$). Since  $\mathscr{C}_c\tuple{C_0,\mathcal{O}}$ is an ontology category, it has an arrow $\pi(E)\longrightarrow \pi(\mathsf{NNF}(\neg F))$  (or $\pi(F)\longrightarrow \pi(\mathsf{NNF}(\neg E))$.   This implies that Properties~(\ref{claim:1})-(\ref{claim:2}) are preserved.
 
    \item \textbf{[$\top$-rule]} and \textbf{[$\top_m$-rule]}  Analogously.
    
    \item \textbf{[$\neg$-rule]}. Assume $\mathscr{C}_c$ has an arrow $C \longrightarrow   D$ and  it has no arrow $\mathsf{NNF}(\neg D) \longrightarrow \mathsf{NNF}(\neg C)$. By the hypothesis on   Properties~(\ref{claim:1}) and (\ref{claim:2}), there is an arrow   $\pi(C)\longrightarrow \pi(D)$. In this case,  \textbf{[$\neg$-rule]} adds to  $\mathscr{C}_c$ an arrow $\mathsf{NNF}(\neg D) \longrightarrow \mathsf{NNF}(\neg C)$. Since $\mathscr{C}_c\tuple{C_0,\mathcal{O}}$ is an ontology category, it has an arrow
    $\pi(\mathsf{NNF}(\neg D))\longrightarrow \pi(\mathsf{NNF}(\neg C))$.   Thus, Properties~(\ref{claim:1})- (\ref{claim:2}) are preserved.   
    
    \item \textbf{[$\sqcup_m$-rule]}. Assume $\mathscr{C}_c$ has an object $E\sqcup F$ and two arrows $E \longrightarrow  X$ and  $F \longrightarrow  X$  and  it has no arrow $E\sqcup F \longrightarrow X$.  By the hypothesis on   Properties~(\ref{claim:1}) and (\ref{claim:2}), there is an object $\pi(E\sqcup F)$ and arrows  $\pi(E) \longrightarrow  \pi(X)$ and  $\pi(F) \longrightarrow  \pi(X)$. In this case,  \textbf{[$\sqcup_m$-rule]} adds to  $\mathscr{C}_c$. Since $\mathscr{C}_c\tuple{C_0,\mathcal{O}}$ is an ontology category,  it has an arrow $\pi(E\sqcup F) \longrightarrow \pi(X)$. Thus, Properties~(\ref{claim:1})- (\ref{claim:2}) are preserved.   
    
    \item \textbf{[$\sqcap_m$-rule]}. Analogously.
    
    \item \textbf{[$\sqcup_\mathsf{dis}$-rule]}. Assume that $\mathscr{C}_c$ has   objects $C\sqcap (D\sqcup E)$,  $W$. It holds that $C\sqcap (D\sqcup E)$ can be rewritten as $C\sqcap (D\sqcup E)=(C_1\sqcup D_1)\sqcap \cdots \sqcap (C_n\sqcup D_n)$. For instance, if $C$ is not a conjunction, then $D,E$ and $C$ take respectively $C_1$, $D_1$, and  $(C_2\sqcup D_2)\sqcap \cdots \sqcap (C_n\sqcup D_n)$. Assume that $\mathsf{check}(\mathscr{C}_c,C\sqcap (D\sqcup E),W)$ returns ``$\mathsf{true}$" (otherwise the rule is not applicable and nothing is changed).
    This implies that for each vector $(X_1,\cdots,X_n)$ with $X_i\in \{C_i,D_i\}$,  $\mathscr{C}_c$ has  an arrow $X_{i_1}\sqcap \cdots \sqcap X_{i_j}\longrightarrow W$. By the hypothesis on   Properties~(\ref{claim:1}), $\mathscr{C}_c\tuple{C_0,\mathcal{O}}$ has  arrows  $\pi(X_{i_1}\sqcap \cdots \sqcap X_{i_j})\longrightarrow \pi(W)$, and thus $X_1\sqcap \cdots \sqcap X_n\longrightarrow \pi(W)$.  It follows that $\mathscr{C}_c\tuple{C_0,\mathcal{O}}$  has  an arrow $\bigsqcup_{X_i\in \{C_i,D_i\}}(X_1\sqcap \cdots \sqcap X_n)\longrightarrow \pi(W)$ due to Definition~\ref{def:disj}, and an arrow
    $\pi(C\sqcap (D\sqcup E))\longrightarrow  \pi(W)$ due to Property~(\ref{conj2}) in Definition~\ref{def:conj}.  In this case, \textbf{[$\sqcup_\mathsf{dis}$-rule]} adds to $\mathscr{C}_c$ an arrows $C\sqcap (D\sqcup E)\longrightarrow W$.     Thus, Properties~(\ref{claim:1}) and (\ref{claim:2}) are preserved.
 
    \item\textbf{[$\exists$-rule]}. Assume that $\mathscr{C}_c$ has  an object $\exists R.D$.   By the hypothesis on   Properties~(\ref{claim:1}) we have 
$\pi(\exists R.D)=\exists R.D$. In this case, \textbf{[$\exists$-rule]} adds to $\mathscr{C}_c$  a role object $R_{(\exists R.D)}$, concept objects $\mathsf{dom}(R_{(\exists R.D)})$, $\mathsf{cod}(R_{(\exists R.D)})$ with arrows $\mathsf{dom}(R_{(\exists R.D)})\longleftrightarrows \exists R.D$. Since  $\mathscr{C}_c\tuple{C_0,\mathcal{O}}$ is an ontology category, it has also  $R_{(\exists R.D)}$, concept objects $\mathsf{dom}(R_{(\exists R.D)})$, $\mathsf{cod}(R_{(\exists R.D)})$ with arrows $\mathsf{dom}(R_{(\exists R.D)})\longleftrightarrows \exists R.D$. We can extend $\pi$ with  $\pi(\mathsf{dom}(R_{(\exists R.D)}))=\mathsf{dom}(R_{(\exists R.D)})$,  $\pi(\mathsf{dom}(R_{(\exists R.D)}))=\mathsf{dom}(R_{(\exists R.D)})$, and $\pi(R_{(\exists R.D)})=R_{(\exists R.D)}$.   Thus, Properties~(\ref{claim:1}) and (\ref{claim:2}) are preserved.
    
    \item\textbf{[$\exists_m$-rule]}. Assume that $\mathscr{C}_c$ has  an object $\exists R.D$ and a role object $R'$ such that $R'\longrightarrow R$ and $\mathsf{cod}(R') \longrightarrow D$. By the hypothesis on   Properties~(\ref{claim:1}) and (\ref{claim:2}) we have 
$\pi(\exists R.D)=\exists R.D$, $\pi(\mathsf{cod}(R'))=\mathsf{cod}(R')$, $\pi(D)=D$ with  $\pi(\mathsf{cod}(R')) \longrightarrow \pi(D)$ and $\pi(R') \longrightarrow \pi(R)$.  In this case, \textbf{[$\exists_m$-rule]} adds to $\mathscr{C}_c$  an arrow $\mathsf{dom}(R') \longrightarrow \mathsf{dom}(R_{(\exists R.D)})$. Since   $\mathscr{C}_c\tuple{C_0,\mathcal{O}}$  is an ontology category, it has $\pi(\mathsf{dom}(R')) \longrightarrow \pi(\mathsf{dom}(R_{(\exists R.D)}))$. Hence, Properties~(\ref{claim:1}) and (\ref{claim:2}) are preserved.
    
    \item \textbf{[$\forall$-rule]}. Assume that $\mathscr{C}_c$ has an object $\forall R.C$, and it has   no arrows $\forall R.C\longleftrightarrows \neg \exists R.\neg C$.  By the hypothesis on   Properties~(\ref{claim:1}) and (\ref{claim:2}) we have 
$\pi(\exists R.D)=\exists R.D$.  In this case, \textbf{[$\exists_m$-rule]} adds to $\mathscr{C}_c$  arrows $\forall R.C\longleftrightarrows \neg \exists R.\neg C$. Since   $\mathscr{C}_c\tuple{C_0,\mathcal{O}}$ is an ontology category, it has and $\pi(\forall R.C)\longleftrightarrows \pi(\neg \exists R.\neg C)$. Hence, Properties~(\ref{claim:1}) and (\ref{claim:2}) are preserved.
    
    \item  \textbf{[$\mathsf{dc}$-rule]}. Analogously.
    
    \item \textbf{[$\mathsf{trans}$-rule]}. Assume $\mathscr{C}_c$ has arrows $C \longrightarrow   D$ and $D \longrightarrow   E$,  and  it has no arrow $C\longrightarrow E$. By the hypothesis on   Properties~(\ref{claim:1}) and (\ref{claim:2}), there is an arrow   $\pi(C)\longrightarrow \pi(D)$. In this case,
\end{enumerate}

To prove the direction ``$\Longleftarrow$" of Property~(\ref{claim:2}) for \textbf{[$\bot$-rule]} and  \textbf{[$\mathsf{trans}$-rule]}, we need an intermediate result. Assume that there is an object of the form $C=(C_1\sqcup D_1) \sqcap  \cdots \sqcap (C_n\sqcup D_n)$. If we apply Property~(\ref{conj2}) in Definition~\ref{def:conj} to $C$ and new objects of the same form resulting from these applications, then we can obtain a tree of objects denoted $T_C$ whose root is $C$, and if a node $C'$ of $T_C$ is of the form $C'=W\sqcap (U\sqcup V)$, then $C'$ has three successors $C'_0=(W\sqcap U) \sqcup (W\sqcap V)$, $C'_1=W\sqcap U$ and $C'_2=W\sqcap V$.     In this case, each object of the form $X_1\sqcap \cdots \sqcap X_n$ is a leaf of $T_C$ with $X_i\in \{C_i,D_i\}$. Note that if $\mathscr{C}_c$ has  an object of the form  $W\sqcap (U\sqcap V)$ and there is no object $(W\sqcap U)\sqcup (W\sqcap V)$, then no saturation rule adds  $(W\sqcap U)\sqcup (W\sqcap V)$ to $\mathscr{C}_c$. Moreover, we have $W\sqcap (U\sqcup V)\longleftrightarrows (W\sqcap U) \sqcup (W\sqcap V)$.    This implies that if there is some $Z\notin \mathsf{Ob}(\mathscr{C}_c)$ and $Z\in \mathsf{Ob}(\mathscr{C}_c\tuple{C_0,\mathcal{O}})$ and $Z$ is not of the form  $(W\sqcap U)\sqcup (W\sqcap V)$, then $Z$ must be a conjunction which is a node of a tree $T_{W\sqcap (U\sqcap V)}$. With such an object  $Z$, we show the following property.
\begin{align*}
& \pi(X)  \longrightarrow Z, Z  \longrightarrow \pi(Y)  \text{ in } \mathscr{C}_c\tuple{C_0,\mathcal{O}} \Longrightarrow X\longrightarrow Y \text{ in } \mathscr{C}_c ~(\dagger)    
\end{align*}
 
First, it is straightforward to show that $(\dagger)$ implies  
 the direction ``$\Longleftarrow$" of Property~(\ref{claim:2}) for    \textbf{[$\mathsf{trans}$-rule]}. We can applies $(\dagger)$ to $\pi(X)\longrightarrow Z \sqcap \mathsf{NNF}(\neg Z)$ to show the direction ``$\Longleftarrow$" of Property~(\ref{claim:2}) for    \textbf{[$\bot$-rule]}. We now show  $(\dagger)$.
 
Indeed, since $Z$ is a conjunction of the  form $C\sqcap D$, we have $\pi(X)  \longrightarrow Z$ is added to $\mathscr{C}_c\tuple{C_0,\mathcal{O}}$ by Definition~\ref{def:conj} if   $\mathscr{C}_c\tuple{C_0,\mathcal{O}}$ has arrows  $\pi(X)\longrightarrow \pi(C)$ and $\pi(X)\longrightarrow \pi(D)$. By  the hypothesis on   Properties~(\ref{claim:1}) and (\ref{claim:2}),  $\mathscr{C}_c$ has arrows $X\longrightarrow C$ and $X\longrightarrow D$. Since $Z\in \mathsf{Ob}(\mathscr{C}_c\tuple{C_0,\mathcal{O}}) \setminus \mathsf{Ob}(\mathscr{C}_c)$, it is not possible that $Z\sqsubseteq \pi(Y)$ is an axiom. Moreover, since $Z$ is an conjunction of the  form $C\sqcap D$, and $Z  \longrightarrow \pi(Y) $ is an arrow of $\mathscr{C}_c\tuple{C_0,\mathcal{O}}$, it follows that either $C=\mathsf{NNF}(\neg D)$, $\pi(Y)=\bot$, or $\mathscr{C}_c\tuple{C_0,\mathcal{O}}$ has an arrow $C  \longrightarrow \pi(Y)$ or $D  \longrightarrow \pi(Y)$.  By  the hypothesis on   Properties~(\ref{claim:1}) and (\ref{claim:2}),  $\mathscr{C}_c$ has an arrow $X\longrightarrow Y$. Hence, $(\dagger)$. is proved.

According to Lemma~\ref{lem:complexity}, Algorithm~\ref{algo:consistency} terminates. 
Due to  Properties~(\ref{claim:1}) and (\ref{claim:2}), $\pi$ is an injection. 
%Hence, $\mathsf{Ob}(\mathscr{C}_c)\subseteq \mathsf{Ob}(\mathscr{C}_c\tuple{C_0,\mathcal{O}})$ and $\mathsf{Hom}(\mathscr{C}_c)\subseteq \mathsf{Hom}(\mathscr{C}_c\tuple{C_0,\mathcal{O}})$.
By hypothesis, $\mathsf{Hom}(\mathscr{C}_c\tuple{C_0,\mathcal{O}})$ contains an arrow $\pi(C_0)\longrightarrow \pi(\bot)$. By definition, the presence of $\pi(C_0)\longrightarrow \pi(\bot)$  must be due to the properties in   Definitions~\ref{def:syntax-cat}-\ref{def:exist}  and Property~(\ref{all-arrow1}) in Definition~\ref{def:forall}. By  Properties~(\ref{claim:1}) and (\ref{claim:2}) on the function $\pi$,  $\mathsf{Hom}(\mathscr{C}_c)$  contains an arrow $C_0\longrightarrow \bot$, and thus Algorithm~\ref{algo:consistency} return ``$\mathsf{false}$". This completes the proof. $\hfill\square$
\end{proof}

The following theorem is a consequence of Lemmas~\ref{lem:complexity} and \ref{lem:sndn-comp}.

\begin{theorem} Satisfiability of  an  $\mathcal{ALC}_{\overline{\forall}}$ concept with respect to an ontology can be decided in   polynomial space.
\end{theorem}
%\begin{proof} According to Lemmas~\ref{lem:sndn-comp} and \ref{lem:complexity}, Algorithm~\ref{algo:consistency} can decide  satisfiability of an $\mathcal{ALC}_{\overline{\forall}}$ concept with respect to an ontology  non-deterministically in polynomial space. By  Savitch's theorem \cite{savitch70,sipser13}, satisfiability of  an  $\mathcal{ALC}_{\overline{\forall}}$ concept with respect to an ontology can be decided in polynomial space.$\hfill\square$
%\end{proof}
 
\begin{example}
Consider the following
$\mathcal{ALC}_{\overline{\forall}}$ ontology :

$
\mathcal{O}
= \{A\sqsubseteq\exists R.C, A\sqsubseteq\forall R.D, D\sqsubseteq\neg C\}$.

Let us check whether $A$ is satisfiable  with respect to $\mathcal{O}$. According to Definition~\ref{def:concept-sat},  $A$ is not satisfiable with respect to $\mathcal{O}$ if there exists  $\mathscr{C}_c\langle{A,\mathcal{O}}\rangle$ that has an arrow $A\longrightarrow \bot$.  For this, we  derive  from the definitions the arrows in the following blocks (I)-(V), and add them to  $\mathscr{C}_c\langle{A,\mathcal{O}}\rangle$:

\begin{center}
\begin{tabular}{|c|l|l|}
\hline
(I)&$A\longrightarrow \exists R.C$& 
$A\longrightarrow \forall R.D$\\ 
&$D\longrightarrow \neg C$& \\
\hline
(II)&$C\sqcap \neg C\longrightarrow \bot$& $D\sqcap  C\longrightarrow   \bot$\\
&$C\longrightarrow \neg D$&\\
\hline
(III)&$\exists R.C\ \longleftrightarrows \mathsf{dom}(R_{(\exists R.C)})$&
$\mathsf{cod}(R_{(\exists R.C)})\longrightarrow C$\\
\hline
(IV)&$\forall R.D \longleftrightarrows \neg \mathsf{dom}(R_{(\exists R.\neg D)})$&$\exists R.\neg D\ \longleftrightarrows \mathsf{dom}(R_{(\exists R.\neg D)})$\\
&$\mathsf{cod}(R_{(\exists R.\neg D)})\longrightarrow \neg D$&\\
\hline
(V)&$\mathsf{dom}(R_{(\exists R. C)})\longrightarrow  \mathsf{dom}(R_{(\exists R.\neg D)})$& $A\longrightarrow \mathsf{dom}(R_{(\exists R.\neg D)})$\\
&$A\longrightarrow \neg \mathsf{dom}(R_{(\exists R.\neg D)})$&
$A\longrightarrow\bot$\\
\hline
\end{tabular}
\end{center}
%\end{table}

%&\neg C\sqcap C\longrightarrow\bot\\
%&\top\longrightarrow C\sqcup\neg C\\
%&\neg\mathsf{cod}(R_{(\exists R.C)})\sqcap\mathsf{cod}(R_{(\exists R.C)})\longrightarrow\bot\\
%&\top\longrightarrow\neg\mathsf{cod}(R_{(\exists R.C)})\sqcup\mathsf{cod}(R_{(\exists R.C)})\\
%&\neg\mathsf{dom}(R_{(\exists R.C)})\sqcap\mathsf{dom}(R_{(\exists R.C)})\longrightarrow\bot\\
%&\top\longrightarrow\neg\mathsf{dom}(R_{(\exists R.C)})\sqcup\mathsf{dom}(R_{(\exists R.C)})\\

Let's check that $\mathscr{C}_c\langle{A,\mathcal{O}}\rangle$  satisfies  Definitions~\ref{def:syntax-cat}-\ref{def:exist} and Property~(\ref{all-arrow1}) in Definition~\ref{def:forall} and it has an arrow $A\longrightarrow \bot$. Indeed,  the arrows in Block~(I) are  added to  $\mathscr{C}_c\langle{A,\mathcal{O}}\rangle$ by Definition~\ref{def:onto-cat}; those in Block~(II)  added   by Definition~\ref{def:neg}; those in Block~(III) added   by Definition~\ref{def:exist}; those in Block~(IV) added   by Definitions~\ref{def:forall} and \ref{def:exist}; and those in Block~(V) added   by Definitions~\ref{def:exist} and \ref{def:neg}. The other arrows   that do not directly contribute to adding $A\longrightarrow \bot$ are omitted.

Now, we apply Algorithm~\ref{algo:consistency} to $(\mathcal{A}, \mathcal{O})$ in order to find an arrow $A\longrightarrow\bot$. We omit most of   
objects and arrows that can be added by  Algorithm~\ref{algo:consistency}, and we focus on those which lead to adding $A\longrightarrow\bot$. %In particular, for the initialisation step, we obtain the first three arrows and then we briefly describe how we obtain the rest.
\begin{center}
\begin{tabular}{|l|l|c|}
\hline
\textbf{Applied rule} & \textbf{Pre-requisite} & \textbf{Arrows (and objects)}\\
& & \textbf{added} to $\Catc$\\
\hline
\textbf{[Initialisation]} & None. & $A\longrightarrow\exists R.C\,(1)$\\ 
\textit{(Line 3 to 5)}& & $A\longrightarrow\forall R.D\,(2)$ \\
& & $D\longrightarrow \neg C\,(3)$ \\

\hline 
\textbf{[$\neg$-rule]} & Arrow (3) is in $\Catc$, 
& $C\longrightarrow\neg D\,(4)$\\  & but no arrow $C\longrightarrow\neg D$. & \\ 

\hline
\textbf{[$\exists$-rule]} & $\exists R.C$ in $\Catc$, & $\mathsf{dom}(R_{(\exists 
R.C)})  \longleftrightarrows\exists R.C\,(5)$\\ & but no object $R_{(\exists R.C)}$ in
$\mathscr{C}_r$. & $\mathsf{cod}(R_{(\exists  R.C)})\longrightarrow C\,(6)$\\ 

\hline 
\textbf{[$\forall$-rule]} &$\forall R.D$ in $\Catc$, &  $\forall R.D 
\longleftrightarrows\neg\exists R.\neg D\,(7)$\\ &  but no arrow $\forall R.D 
\longleftrightarrows\neg\exists R.\neg D$.& $\exists  R.\neg D\,(8)$\\ 

\hline  
\textbf{[$\exists$-rule]} & $\exists R.\neg D$ in $\Catc$, &  
$\mathsf{dom}(R_{(\exists R.\neg D)})\longleftrightarrows \exists R.\neg D\,(9)$\\
& but no object $R_{(\exists R.\neg D)}$ in $\mathscr{C}_r$.& 
    $\mathsf{cod}(R_{(\exists R.\neg D)})\longrightarrow\neg D\,(10)$\\
\hline

\textbf{[$\mathsf{trans}$-rule]} & We have arrows (4) and (6). & 
$\mathsf{cod}(R_{(\exists R.C)})\longrightarrow\neg D\,(11)$\\

\hline \textbf{[$\exists_m$-rule]} & We have arrow (11) and arrow & 
$\mathsf{dom}(R_{(\exists R.C)})\longrightarrow$ \\
&  $R_{(\exists R.C)}\longrightarrow R$, but no arrow & $\mathsf{dom}(R_{(\exists R.\neg D)})\,(12)$ \\
& $\mathsf{dom}(R_{(\exists R.C)})\longrightarrow\mathsf{dom}(R_{(\exists 
    R.\neg D)})$. & \\
    
\hline
\textbf{[$\mathsf{trans}$-rule]}& From arrow (1) to (5) and (1) & 
       $A\longrightarrow\mathsf{dom}(R_{(\exists R.\neg D)})\,(13)$ \\
& to (12) successively.& \\

\hline
\textbf{[$\mathsf{trans}$-rule]} & From arrow (2) to (7).& $A\longrightarrow
\neg\exists R.\neg D$\\
+\textup{[$\neg$-rule]} & On arrow (9).& $\neg\exists R.\neg D\longrightarrow
\neg\mathsf{dom}(R_{(\exists R.\neg D)})$\\
+\textbf{[$\mathsf{trans}$-rule]}& From arrow (2) to last arrow.& 
    $A\longrightarrow\neg\mathsf{dom}(R_{(\exists R.\neg D)})\,(14)$\\
    
\hline 
\textbf{[$\bot$-rule]} & We have arrow (13) and (14), & $A\longrightarrow\bot
\,(15)$ \\
& but no arrow $A\longrightarrow\bot$. & \\

\hline
\end{tabular}
\end{center}
The rules never remove any arrow, thus the algorithm will always be able 
to find (line 9) the arrow $A\longrightarrow\bot$ if it is added, hence 
the algorithm return false. Since all rules are deterministic, there is 
no category without $A\longrightarrow\bot$ that can be constructed.
\end{example}

\begin{example}
Consider the following $\mathcal{ALC}_{\overline{\forall}}$ concept $C_0 = (A\sqcup B)
\sqcap(C\sqcup D)\sqcap(E\sqcup F)$ and ontology :
\[
\mathcal{O}
= \{A\sqcap C\sqsubseteq\bot, A\sqcap D\sqsubseteq\bot, B\sqcap C\sqsubseteq\bot,
B\sqcap D\sqsubseteq\bot\}
\]
Let us check whether $C_0$ is satisfiable  with respect to $\mathcal{O}$. As for the previous
example, we have from Definition~\ref{def:concept-sat} that $C_0$ is not satisfiable with respect to $\mathcal{O}$ if there exists  $\mathscr{C}_c\langle{C_0,\mathcal{O}}\rangle$ that has an arrow $C_0\longrightarrow\bot$. We derive  from the definitions the arrows in the following blocks (I)-(V), and add them to  $\mathscr{C}_c\langle{C_0,\mathcal{O}}\rangle$. We rename a particular object for the sake
of clarity, it comes from Property~(\ref{conj2}) of 
Definition~\ref{def:conj} and is  a notation :
\begin{align*}
C_1 =&(A\sqcap C\sqcap E)\sqcup(B\sqcap C\sqcap E)\sqcup(A\sqcap D\sqcap E)\\
&\sqcup(B\sqcap D\sqcap E)\sqcup(A\sqcap C\sqcap F)\sqcup(B\sqcap C\sqcap F)\\
&\sqcup(A\sqcap D\sqcap F)\sqcup(B\sqcap D\sqcap F)
\end{align*}
we omit the intermediate arrows for the sake of clarity.
%\begin{table}
%\centering
\begin{center}
\begin{tabular}{|c|l|l|l|}
\hline
(I)&$A\sqcap C\longrightarrow\bot$ & $A\sqcap D\longrightarrow\bot$ & $B\sqcap 
C\longrightarrow\bot$\\ 
& $B\sqcap D\longrightarrow\bot$ &  &\\
\hline
(II)& $C\longrightarrow\neg A$ & $D\longrightarrow\neg A$ & $C\longrightarrow\neg B$\\
& $D\longrightarrow\neg B$ &  & \\
\hline
(III)&$A\longrightarrow A\sqcup B$ & $B\longrightarrow A\sqcup B$ & 
$D\longrightarrow C\sqcup D$\\
& $C\longrightarrow C\sqcup D$ & $F\longrightarrow E\sqcup F$ & 
$E\longrightarrow E\sqcup F$\\
& $A\sqcap C\sqcap E\longrightarrow C_1$ & $B\sqcap C\sqcap E\longrightarrow C_1$ & 
$A\sqcap D\sqcap E\longrightarrow C_1$\\
& $A\sqcap C\sqcap F\longrightarrow C_1$ & $B\sqcap D\sqcap E\longrightarrow C_1$ & 
$B\sqcap C\sqcap F\longrightarrow C_1$\\
& $A\sqcap D\sqcap F\longrightarrow C_1$ & $B\sqcap D\sqcap F\longrightarrow C_1$ & 
$C_1\longrightarrow\bot$\\
\hline
(IV)& $C_0\longrightarrow A\sqcup B$ & $C_0\longrightarrow C\sqcup D$ & 
$C_0\longrightarrow E\sqcup F$\\
& $C_0\longrightarrow C_1$ & $A\sqcap C\sqcap E\longrightarrow A\sqcap C$
& $B\sqcap C\sqcap E\longrightarrow B\sqcap C$\\
& $A\sqcap D\sqcap E\longrightarrow A\sqcap D$ & $B\sqcap D\sqcap E\longrightarrow B\sqcap D$
& $A\sqcap C\sqcap F\longrightarrow A\sqcap C$\\
& $B\sqcap C\sqcap F\longrightarrow B\sqcap C$ & $A\sqcap D\sqcap F\longrightarrow A\sqcap D$
& $B\sqcap D\sqcap F\longrightarrow B\sqcap D$\\
\hline
(V)  & $A\sqcap C\sqcap E\longrightarrow\bot$ & $B\sqcap C\sqcap E\longrightarrow\bot$
&$A\sqcap D\sqcap E\longrightarrow\bot$ \\
& $B\sqcap D\sqcap E\longrightarrow\bot$ & $A\sqcap C\sqcap F\longrightarrow\bot$
& $B\sqcap C\sqcap F\longrightarrow\bot$ \\
& $A\sqcap D\sqcap F\longrightarrow\bot$ & $B\sqcap D\sqcap F\longrightarrow\bot$ 
& $C_0\longrightarrow\bot$\\
\hline
\end{tabular}
\end{center}
%\end{table}
%&\neg C\sqcap C\longrightarrow\bot\\
%&\top\longrightarrow C\sqcup\neg C\\
%&\neg\mathsf{cod}(R_{(\exists R.C)})\sqcap\mathsf{cod}(R_{(\exists R.C)})\longrightarrow\bot\\
%&\top\longrightarrow\neg\mathsf{cod}(R_{(\exists R.C)})\sqcup\mathsf{cod}(R_{(\exists R.C)})\\
%&\neg\mathsf{dom}(R_{(\exists R.C)})\sqcap\mathsf{dom}(R_{(\exists R.C)})\longrightarrow\bot\\
%&\top\longrightarrow\neg\mathsf{dom}(R_{(\exists R.C)})\sqcup\mathsf{dom}(R_{(\exists R.C)})\\
Let's check that $\mathscr{C}_c\langle{C_0,\mathcal{O}}\rangle$  satisfies  Definitions~\ref{def:syntax-cat}-\ref{def:exist} and Property~(\ref{all-arrow1}) in Definition~\ref{def:forall} and it has an arrow $C_0\longrightarrow \bot$. Indeed,  the arrows in Block~(I) are  added to  $\mathscr{C}_c\langle{C_0,\mathcal{O}}\rangle$ by Definition~\ref{def:onto-cat}; those in Block~(II)  added   by Definition~\ref{def:neg}; those in Block~(III) added   by Definition~\ref{def:disj}; those in Block~(IV) added   by Definitions~\ref{def:conj}; and those in Block~(V) added   by Definitions~\ref{def:syntax-cat}.

As before, we apply Algorithm~\ref{algo:consistency} to $(C_0, \mathcal{O})$ in order to find an arrow $C_0\longrightarrow\bot$. %In particular, for the initialisation step, we obtain the first three arrows and then we briefly describe how we obtain the rest.

\begin{center}
\begin{tabular}{|l|l|c|}
\hline
\textbf{Applied rule} & \textbf{Pre-requisite} & \textbf{Arrows added to $\Catc$}\\
\hline
\textbf{[Initialisation]} & None. & $A\sqcap C\longrightarrow\bot\,(1)$\\ 
\textit{(Line 3 to 5)}& & $B\sqcap C\longrightarrow\bot\,(2)$ \\
& & $A\sqcap D\longrightarrow\bot\,(3)$ \\
& & $B\sqcap D\longrightarrow\bot\,(4)$ \\
\hline 
\hline
\textbf{Algorithm~\ref{algo:check}} & $\mathsf{check}(\Catc, C_0, \bot)$ & $C_0 = (A\sqcup 
B)\sqcap(C\sqcup D)\sqcap(E\sqcup F)$\\
\hline
1st loop & $A\sqcap C\sqcap E$ & (1), \textbf{[$\sqcap$-rule]} and 
\textbf{[$\mathsf{trans}$-rule]}\\
& & imply $A\sqcap C\sqcap E\longrightarrow\bot$\\
\hline
2nd loop & $B\sqcap C\sqcap E$ & (2), \textbf{[$\sqcap$-rule]} and 
\textbf{[$\mathsf{trans}$-rule]}\\
& & imply $B\sqcap C\sqcap E\longrightarrow\bot$\\
\hline
3rd loop & $A\sqcap D\sqcap E$ & (3), \textbf{[$\sqcap$-rule]} and 
\textbf{[$\mathsf{trans}$-rule]}\\
& & imply $A\sqcap D\sqcap E\longrightarrow\bot$\\
\hline
4th loop & $B\sqcap D\sqcap E$ & (4), \textbf{[$\sqcap$-rule]} and 
\textbf{[$\mathsf{trans}$-rule]}\\
& & imply $B\sqcap D\sqcap E\longrightarrow\bot$\\
\hline
5th loop & $A\sqcap C\sqcap F$ & (1), \textbf{[$\sqcap$-rule]} and 
\textbf{[$\mathsf{trans}$-rule]}\\
& & imply $A\sqcap C\sqcap F\longrightarrow\bot$\\
\hline
6th loop & $B\sqcap C\sqcap F$ & (2), \textbf{[$\sqcap$-rule]} and 
\textbf{[$\mathsf{trans}$-rule]}\\
& & imply $B\sqcap C\sqcap F\longrightarrow\bot$\\
\hline
7th loop & $A\sqcap D\sqcap F$ & (3), \textbf{[$\sqcap$-rule]} and 
\textbf{[$\mathsf{trans}$-rule]}\\
& & imply $A\sqcap D\sqcap F\longrightarrow\bot$\\
\hline
8th loop & $B\sqcap D\sqcap F$ & (4), \textbf{[$\sqcap$-rule]} and 
\textbf{[$\mathsf{trans}$-rule]}\\
& & imply $B\sqcap D\sqcap F\longrightarrow\bot$\\
\hline
End loop & For each $1\leq i\leq 3$ with & \\
& and  $X_i\in\{A, B\}, Y_i\in\{C,D\},$ & $\mathsf{check}$ returns\\
&  $Z_i\in\{E,D\}$, we have an arrow& \textbf{true}\\
&  $X_i\sqcap Y_i\sqcap Z_i\longrightarrow\bot$.&\\
\hline
\hline
\textbf{[$\sqcup_{\mathsf{dis}}$-rule]} & There is no arrow $C_0\longrightarrow\bot$ and 
& $C_0\longrightarrow\bot$\\
& $\mathsf{check}(\Catc, C_0, \bot)$ returns $\textbf{true}$. & \\
\hline
\end{tabular}
\end{center}

By the same arguments as the previous example, there is no category without 
$C_0\longleftrightarrows\bot$ that can be constructed, hence $C_0$ is unstatisfiable.
\end{example}
 
\section{Conclusion} We have presented a rewriting of the usual set-theoretical semantics of $\mathcal{ALC}$  by using categorial language.   Thanks to the modular representation of  category-theoretical semantics  composed of separate constraints, we identify a sublogic  of $\mathcal{ALC}$, namely $\mathcal{ALC}_{\overline{\forall}}$, and show that it is strictly from $\mathcal{ALC}$. We also proposed a {\sc{PSPACE}} non-deterministic algorithm  for checking  concept unsatisfiability  in  $\mathcal{ALC}_{\overline{\forall}}$, which implies that $\mathcal{ALC}_{\overline{\forall}}$ is {\sc{PSPACE}}. For future work, we will investigate the question whether $\mathcal{ALC}_{\overline{\forall}}$ is   {\sc{PSPACE}}-complete. This question is open because $\mathcal{ALC}_{\overline{\forall}}$ (with general TBoxes) may not be included in  $\mathcal{ALC}$ without TBox. Moreover, we believe that category-theoretical semantics can be extended to more expressive DLs with role constructors. For instance, role functionality can be expressed as \emph{monic and epic arrows} \cite{gol06}.

\bibliographystyle{elsarticle-num} \bibliography{main}
\end{document}